\documentclass[11pt]{article}

\usepackage[letterpaper, margin=1in]{geometry}
\usepackage[utf8]{inputenc}
\usepackage[T1]{fontenc}

\usepackage{amsmath,amssymb,amsthm,mathtools,bm}
\usepackage{booktabs,array,longtable}
\usepackage{graphicx}
\usepackage{multirow}
\usepackage{ragged2e}
\usepackage{enumitem}
\usepackage{natbib}
\usepackage[dvipsnames]{xcolor}
\usepackage{tikz}
\usetikzlibrary{arrows.meta,positioning,fit,shapes.geometric}
\usepackage[colorlinks=true,citecolor=MidnightBlue,linkcolor=MidnightBlue,urlcolor=MidnightBlue]{hyperref}

\graphicspath{{numerical_figures/}}

\newcommand{\E}{\mathbb{E}}
\newcommand{\Pp}{\mathbb{P}}
\newcommand{\R}{\mathbb{R}}

\newcommand{\F}{\mathcal{F}}
\newcommand{\G}{\mathcal{G}}
\newcommand{\D}{\mathcal{D}}
\newcommand{\A}{\mathcal{A}}

\newcommand{\given}{\mid}
\newcommand{\iid}{\stackrel{\mathrm{iid}}{\sim}}
\newcommand{\ind}{\mathbf{1}}

\newcommand{\expit}{\operatorname{expit}}
\newcommand{\argmax}{\operatorname*{arg\,max}}

\newcommand{\norm}[1]{\left\lVert #1 \right\rVert}

\theoremstyle{plain}
\newtheorem{theorem}{Theorem}[section]
\newtheorem{proposition}[theorem]{Proposition}
\newtheorem{corollary}[theorem]{Corollary}

\theoremstyle{definition}
\newtheorem{definition}[theorem]{Definition}
\newtheorem{example}[theorem]{Example}

\theoremstyle{remark}
\newtheorem{remark}[theorem]{Remark}

\title{\bfseries Target-Oriented Statistical Compression: Sufficiency, Reverse Martingales, and Sequential Monitoring
}

\author{
Yuan-chin Ivan Chang\\[3pt]
\small Institute of Statistical Science, Academia Sinica\\
\small 128 Academia Road, Section 2, Nankang, Taipei 11529, Taiwan\\
\small \href{mailto:ycchang@as.edu.tw}{ycchang@as.edu.tw}
}

\date{\today}

\begin{document}

\maketitle

\begin{abstract}
Statistical procedures rarely use all available data: a sufficient statistic
discards features irrelevant to a parameter; a maximum likelihood estimate
compresses an empirical objective into an optimizing point; a hidden state in a
sequential model compresses a history into a learned representation.  We
develop a unified account of these practices as \emph{target-oriented
statistical compression} and use it to motivate a novel three-condition
sequential stopping rule, $\tau_{\mathrm{RM}}$, for practical boundary
declarations in binary monitoring problems.

The central object is the conditional target process $M_n=\E(Z\mid\G_n)$,
where $Z$ is the inferential target and $\G_n=\sigma(T_n)$ is the information
retained by the compression map $T_n$.  Arranging $(\G_n)$ as a decreasing
filtration makes $(M_n)$ a reverse martingale converging to
$M_\infty=\E(Z\mid\G_\infty)$.  Approximate summaries produce reverse
quasi-martingale defects $\delta_n$ quantifying information loss; the
observable $r_n=|M_n-M_{n-1}|$ serves as a stability proxy.  The rule
\[
  \tau_{\mathrm{RM}}=\inf\bigl\{n\ge n_{\min}:
  B_n\le\varepsilon,\; W_n\le w,\; r_n\le\eta\bigr\}
\]
requires boundary closeness, uncertainty localization, and trajectory
stability simultaneously.  When the retained summary is exactly sufficient,
$\delta_n=0$ and $\tau_{\mathrm{RM}}$ reduces automatically to the
two-condition form.  Simulation studies across Bernoulli, Gaussian, Poisson,
and logistic settings confirm substantial reductions in false boundary
declaration rates relative to boundary-only and two-condition alternatives.
\end{abstract}



\section{Introduction}

Statistical procedures rarely use the full observed path.  A sample mean forgets
order, extremes, and local fluctuations; a sufficient statistic discards
features irrelevant to a parameter under a specified model; a maximum likelihood
estimate compresses an empirical objective into a single optimizing parameter;
a hidden state in a sequential prediction system compresses a history into a
learned representation.  These examples share one principle:
{\it good statistical summaries remember what matters for the target,}
not necessarily every detail of the realized data path.

\paragraph{The compression triple and its stopping-rule consequence.}
This article develops that principle as a theory of target-oriented
information compression, and uses it to motivate a sequential methodology
for practical boundary declarations.  The key object is the triple
\[
  (T_n,\;\G_n=\sigma(T_n),\;M_n=\E(Z\mid\G_n)),
\]
where $T_n$ is a data-compression map, $\G_n$ records the retained
information (see \eqref{eq:target_projection}), and $M_n$ is the conditional
target process.  When $(\G_n)$ forms a decreasing filtration, $M_n$ is a
reverse martingale converging to $M_\infty=\E(Z\mid\G_\infty)$.  A run of all
failures in Bernoulli data, a separated logistic regression, or a near-zero
sequential risk score may tempt an exact statement such as ``the event
probability is zero.''  Finite data justify no such claim.  The distinction is
fundamental: exact boundary degeneracy is a property of $M_\infty$, whereas
finite boundary closeness is only a diagnostic.

The practical consequence is the three-condition stopping rule
\begin{equation}\label{eq:tau_rm_intro}
  \tau_{\mathrm{RM}}
  =\inf\bigl\{n\ge n_{\min}:
  B_n\le\varepsilon,\;W_n\le w,\;r_n\le\eta\bigr\},
\end{equation}
where $B_n=\min\{M_n,1-M_n\}$ measures boundary closeness, $W_n$ is the
width of a time-uniform confidence sequence or posterior credible interval for
$M_\infty$, and $r_n=|M_n-M_{n-1}|$ is the empirical stability proxy for the
quasi-martingale defect~\eqref{eq:quasi_defect}.  Compared with standard
sequential methods, $\tau_{\mathrm{RM}}$ addresses a complementary question:
the SPRT \citep{wald1945} tests a simple-versus-simple hypothesis about the
mean, and CUSUM \citep{siegmund1985} detects a sustained shift, while
$\tau_{\mathrm{RM}}$ asks whether the conditional target is stably near a
boundary, integrating closeness, uncertainty, and coherence into one criterion.

\paragraph{Exact versus approximate compression.}
When the retained summary is exactly sufficient, the coherence defect
$\delta_n=\E(M_n\mid\G_{n+1})-M_{n+1}$ is identically zero and
Proposition~\ref{prop:exact_reverse_coherence} guarantees that
$\tau_{\mathrm{RM}}\equiv\tau_{2\mathrm{cond}}$, so the stability screen
imposes no delay.  When the summary is only approximately sufficient ---
penalized estimators, risk scores, learned representations --- $\delta_n\ne0$
and $r_n\le\eta$ provides genuine additional protection against premature
declarations.

\paragraph{New contributions.}
The applied stopping procedures, finite-sample error bounds, and Bernoulli,
logistic, and RAND~HIE numerical evidence are developed in the companion paper
\citet{chang2025rm}.  The present paper contributes the compression triple as a
unified organizing principle; the tail-$\sigma$-field construction of decreasing
filtrations (Remark~\ref{rem:decreasing_filtration}); the reverse quasi-martingale
defect $\delta_n$ and its connection with $r_n$ (Remark~\ref{rem:delta_to_r});
the structural reduction $\tau_{\mathrm{RM}}\equiv\tau_{2\mathrm{cond}}$ for
exactly sufficient summaries (Proposition~\ref{prop:exact_reverse_coherence});
and numerical calibration across Gaussian, Poisson, public-health surveillance,
and quasi-martingale settings not treated in the companion paper
(Section~\ref{sec:discussion}).  Section~\ref{sec:classical} recasts the Normal,
Bernoulli, and Poisson sufficiency examples as lossless-compression anchors;
the algebraic derivations are in Appendix~\ref{app:sufficiency_derivations}.

\paragraph{Organisation.}
Section~\ref{sec:compression} develops the mathematical foundation of
target-oriented compression, including the reverse quasi-martingale defect and
the decreasing-filtration construction.  Section~\ref{sec:classical} recasts
exact sufficiency and maximum likelihood as lossless and stabilized empirical
compression.  Section~\ref{sec:rm_boundary} applies the framework to boundary
degeneracy in sequential binary problems and defines the three-condition stopping
rule $\tau_{\mathrm{RM}}$ (Section~\ref{sec:stopping}), including its
error-control guarantee (Corollary~\ref{cor:error_control}) and parameter
selection guidance (Remark~\ref{rem:tuning}).  Section~\ref{sec:discussion}
presents numerical calibration studies.  Proofs and data details are in the
Appendix and Supplementary Material.

\section{Target-oriented compression: mathematical foundation}
\label{sec:compression}

Let \(X_1,\ldots,X_n\) denote observed data, let
\(T_n=T_n(X_1,\ldots,X_n)\) be a statistic (the compression map), and let
\(\G_n=\sigma(T_n)\) be its associated retained information field.
For a target quantity \(Z\in L^1\) the corresponding target projection is
\begin{equation}
\label{eq:target_projection}
    M_n=\E(Z\given\G_n),
\end{equation}
where \(\G_n=\sigma(T_n)\) or a related coarsened \(\sigma\)-field.
This ordering is important.  The statistic determines what is remembered; the
\(\sigma\)-field records the information retained; the conditional expectation
projects the target onto that information.

The full path contains ordering, local jumps, extremes, clusters, and other
details.  The statistic keeps only selected information.  The value of that
compression depends on the target.  If the target is the population mean, the
sample mean
\[
    \bar X_n=\frac1n\sum_{i=1}^n X_i
\]
is a natural compression.  It discards order and much of the pathwise shape but
preserves central level.  If the target is variance, \(\bar X_n\) alone is
insufficient and one must also keep a spread summary such as
\[
    S_n^2=\frac{1}{n-1}\sum_{i=1}^n(X_i-\bar X_n)^2.
\]
If the target is distributional shape, the empirical distribution function
\[
    F_n(t)=\frac1n\sum_{i=1}^n\ind\{X_i\leq t\}
\]
is a richer compression than the mean.  If the target is a pairwise or
rank-based functional, a U-statistic
\[
    U_n=\binom{n}{m}^{-1}
    \sum_{1\leq i_1<\cdots<i_m\leq n}
    h(X_{i_1},\ldots,X_{i_m})
\]
may be the relevant compressed object.

\begin{table}[t]
\centering
\caption{Examples of target-oriented statistical compression.}
\label{tab:compression}
\begin{tabular}{lll}
\toprule
Model or target & Natural compression & Information retained\\
\midrule
Normal mean, known variance & \(\bar X_n\) & central level\\
Normal mean and variance & \((\bar X_n,S_n^2)\) & level and spread\\
Bernoulli probability & \(\sum_iX_i\) & success count\\
Distribution function & \(F_n\) & empirical shape\\
Pairwise functional & \(U_n\) & kernel average\\
Sequential prediction & hidden state \(H_t\) & task-relevant history\\
\bottomrule
\end{tabular}
\end{table}

Sequential observation is usually described by an increasing filtration
\(\F_n^{\mathrm{raw}}=\sigma(X_1,\ldots,X_n)\), with
\(\F_1^{\mathrm{raw}}\subseteq\F_2^{\mathrm{raw}}\subseteq\cdots\).
This is a correct and natural viewpoint for data arrival.  However, it is not always the correct
viewpoint for compression.  The sample mean, for example, is generally not a
martingale with respect to the raw-data filtration because
\[
    \E(\bar X_{n+1}\given \F_n^{\mathrm{raw}})
    =
    \frac{n\bar X_n+\mu}{n+1},
\]
which is not equal to \(\bar X_n\) unless \(\bar X_n=\mu\).

A compression viewpoint instead uses decreasing information,
\(\G_1\supseteq\G_2\supseteq\G_3\supseteq\cdots\), where later
\(\sigma\)-fields contain less information.  This formalizes the idea of
moving from more detailed summaries to coarser summaries.

\begin{remark}[Constructing a decreasing filtration from sequential statistics]
\label{rem:decreasing_filtration}
For a single statistic \(T_n=T_n(X_1,\ldots,X_n)\), the \(\sigma\)-field
\(\sigma(T_n)\) is not automatically contained in \(\sigma(T_{n-1})\);
richer data at step \(n\) can make \(\sigma(T_n)\) \emph{larger} than
\(\sigma(T_{n-1})\).  The standard construction that guarantees a decreasing
filtration is the \emph{tail \(\sigma\)-field}:
\begin{equation}
    \G_n = \sigma(T_k : k\geq n),
    \qquad n\geq1.
\end{equation}
Here \(\G_n\) is generated by all summaries from step \(n\) onward, so
\(\G_1\supseteq\G_2\supseteq\cdots\) holds by construction.  Equivalently,
one can work with a fixed terminal horizon \(N\) and define
\(\G_n=\sigma(T_n,T_{n+1},\ldots,T_N)\).  When \(T_n\) is an exact
sufficient statistic (such as the running mean or running proportion) whose
\(\sigma\)-field equals that of the full data up to \(n\), the tail construction
reduces to the natural backward filtration.  In approximate settings (penalized
estimators, hidden states), the same tail construction applies and a coherence
defect (defined formally as \(\delta_n\) in \eqref{eq:quasi_defect} below)
quantifies the departure from the exact case.
\end{remark}

\begin{definition}[Reverse martingale]
Let \((\G_n)_{n\geq1}\) be a decreasing filtration.  An integrable,
\(\G_n\)-measurable sequence \((M_n)\) is a reverse martingale if
\[
    \E(M_n\given \G_{n+1})=M_{n+1},
    \qquad n\geq1.
\]
\end{definition}

\begin{theorem}[Conditional target reverse martingale]
\label{thm:conditional_rm}
Let \(Z\in L^1\), let \(\G_1\supseteq\G_2\supseteq\cdots\) be a decreasing
filtration, and define \(M_n=\E(Z\given\G_n)\) as in \eqref{eq:target_projection}.
Then \((M_n,\G_n)\) is a reverse martingale.  Moreover, if
\(\G_\infty=\cap_{n\geq1}\G_n\), then
\[
    M_n\longrightarrow M_\infty=\E(Z\given\G_\infty)
\]
almost surely and in \(L^1\).
\end{theorem}

\begin{proof}
Since \(\G_{n+1}\subseteq\G_n\), the tower property gives
\[
    \E(M_n\given\G_{n+1})
    =
    \E\{\E(Z\given\G_n)\given\G_{n+1}\}
    =
    \E(Z\given\G_{n+1})
    =
    M_{n+1}.
\]
The convergence assertion is the reverse-martingale convergence theorem
\citep{ville1939,robbins1970,doob1953,bjork1996,williams1991,kallenberg2002,durrett2019}.
\end{proof}

This theorem says that the reverse-martingale structure belongs naturally to
the conditional target process.  It does not say that every statistic is a
reverse martingale.

\begin{remark}[Two interpretations of the target \(Z\), and a unified uncertainty metric]
\label{rem:target_types}
The target $Z\in L^1$ admits two structurally distinct interpretations.

\begin{enumerate}[label=(\alph*)]
\item \textbf{Frequentist predictive (Case~A).}
      Set $Z=Y_{n+1}$ or $Z=h(X_{n+1},X_{n+2},\ldots)$.  Then
      $M_n=\E(Z\mid\G_n)$ is a sequence of predictive means converging to the
      long-run conditional prediction $M_\infty$.  No prior is needed.
      The uncertainty width $W_n$ is the width of a time-uniform confidence
      sequence for $M_\infty$; under the \citet{howard2021} construction,
      $W_n\asymp C\sqrt{(\log n)/n}$ for a constant $C$ depending on the
      variance proxy.  For bounded targets $M_n\in[0,1]$ (the Bernoulli and
      logistic settings of our numerical studies), the betting-based confidence
      sequences of \citet{waudbysmith2023} yield sharper widths and are
      directly applicable.

\item \textbf{Bayesian posterior contraction (Case~B).}
      Set $Z=\theta$ with a prior.  Then $M_n=\E(\theta\mid\G_n)$ is the
      posterior mean and $M_\infty$ is its limiting value.  Under a conjugate
      model with a sufficient statistic, the posterior mean is a deterministic
      function of that statistic and the reverse-martingale convergence encodes
      posterior contraction \citep{fong2023}.  Here $W_n$ is the width of the
      $(1-\alpha)$ posterior credible interval.
\end{enumerate}

The two interpretations share the same $\tau_{\mathrm{RM}}$ scorecard because
both $W_n$ sequences are \emph{uncertainty widths that shrink to zero under
the correct model} --- whether from a confidence-sequence or a posterior
contraction argument.  Concretely, for Bernoulli data with a Jeffreys prior
(Case~B), $W_n$ is the Beta credible interval width $\approx
2z_{\alpha/2}/\sqrt{n+1}$, which matches the confidence-sequence rate
$O(n^{-1/2})$ from Case~A.  The stopping rule therefore applies uniformly
across both frameworks; only the formula for $W_n$ differs.
The numerical studies use Jeffreys posterior means for Bernoulli and Poisson
(Case~B) and ridge-MLE predicted probabilities for logistic regression (Case~A).
\end{remark}

\subsection{Approximate sufficiency and reverse quasi-martingales}

Modern summaries are often not exactly sufficient.  Examples include selected
features, principal components, penalized estimates, risk scores, and hidden
states in recurrent or transformer-based prediction systems
\citep{hochreiter1997,goodfellow2016}.  They may retain
most target-relevant information without retaining it exactly.

Let \(M_n\) be a target process adapted to a decreasing filtration
\((\G_n)\).  Define the reverse quasi-martingale defect
\begin{equation}
\label{eq:quasi_defect}
    \delta_n
    =
    \E(M_n\given\G_{n+1})-M_{n+1},
\end{equation}
where the absolute summability criterion \(\sum_{n=1}^\infty\E|\delta_n|<\infty\)
characterizes the reverse quasi-martingale class.
If \(\delta_n=0\) for all \(n\), the exact reverse-martingale property is
recovered.  If the defects satisfy the summability condition together with
standard integrability conditions, then the process behaves like a reverse
quasi-martingale and retains a convergence theory analogous to stable compression.  
The defect \(\delta_n\) measures the failure of exact
coherence across successive compression levels.  It is therefore a quantitative
measure of information loss or representation mismatch.

\begin{remark}[From theoretical defect \(\delta_n\) to empirical diagnostic \(r_n\)]
\label{rem:delta_to_r}
The observable step-to-step change \(r_n=|M_n-M_{n-1}|\) is the operational
counterpart of the theoretical defect \(\delta_n\) defined in
\eqref{eq:quasi_defect}.  Under the quasi-martingale summability condition
\(\sum_{n\geq1}\E|\delta_n|<\infty\), the reverse quasi-martingale convergence
theorem guarantees \(M_n\to M_\infty\) almost surely, and consequently
\(r_n\to0\) almost surely.  This justifies using $r_n$ as a stopping-rule
component: persistent large values signal either a large defect $\delta_n$
(information leakage) or a slowly converging path (uncertainty not yet
resolved).  When the summability condition fails, $r_n$ need not tend to zero
and the stopping rule correctly refuses to fire.

In contrast to classical boundary-detection tools, $r_n$ plays a different
role.  A CUSUM chart \citep{siegmund1985} accumulates evidence of a mean shift
relative to a reference value; the SPRT \citep{wald1945} computes the
log-likelihood ratio for a pre-specified pair of point hypotheses.  Neither
tool assesses whether the \emph{conditional target} has stabilized near a
boundary.  The $r_n$ component does precisely this: it is the step-to-step
stability screen for the target projection $M_n$, not for the raw data mean.
The three-condition scorecard~\eqref{eq:tau_rm} therefore complements rather
than competes with SPRT and CUSUM (see Section~\ref{subsec:exact} for
numerical comparisons).
\end{remark}

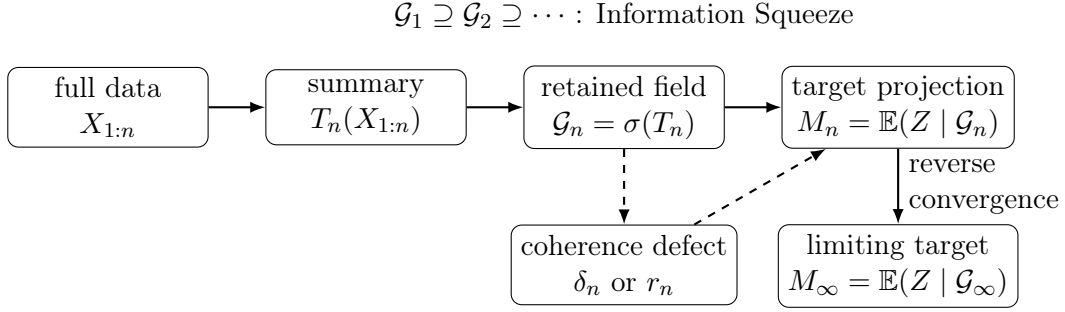
\begin{figure}[t]
\centering
\begin{tikzpicture}[
    node distance=0.75cm,
    box/.style={draw, rounded corners, align=center, minimum width=2.65cm, minimum height=0.9cm},
    smallbox/.style={draw, rounded corners, align=center, minimum width=2.35cm, minimum height=0.75cm},
    arrow/.style={-{Latex[length=2mm]}, thick}
]
\node[box] (data) {full data\\ \(X_{1:n}\)};
\node[box, right=of data] (stat) {summary\\ \(T_n(X_{1:n})\)};
\node[box, right=of stat] (field) {retained field\\ \(\G_n=\sigma(T_n)\)};
\node[box, right=of field] (target) {target projection\\ \(M_n=\E(Z\mid\G_n)\)};
\node[box, below=1.0cm of target] (limit) {limiting target\\ \(M_\infty=\E(Z\mid\G_\infty)\)};
\node[smallbox, below=1.0cm of field] (defect) {coherence defect\\ \(\delta_n\) or \(r_n\)};
\draw[arrow] (data) -- (stat);
\draw[arrow] (stat) -- (field);
\draw[arrow] (field) -- (target);
\draw[arrow] (target) -- node[right, align=left] {reverse\\ convergence} (limit);
\draw[arrow, dashed] (field) -- (defect);
\draw[arrow, dashed] (defect) -- (target);
\node[align=center, above=0.35cm of field] {\(\G_1\supseteq\G_2\supseteq\cdots\)~: Information Squeeze};
\end{tikzpicture}
\caption{The information squeeze: data are compressed into a statistic, the
statistic induces a retained \(\sigma\)-field, and the target is projected onto that
field.  Exact reverse martingales have zero coherence defect; approximate
summaries produce theoretical defects \(\delta_n\) and empirical diagnostics
\(r_n\).}
\label{fig:information_squeeze}
\end{figure}

A compact summary is
\[
\begin{aligned}
\text{full data}
&\longrightarrow
\text{sufficient or approximate statistic}\\
&\longrightarrow
\text{conditional target process}\\
&\longrightarrow
\text{reverse martingale or reverse quasi-martingale}.
\end{aligned}
\]

\subsection{Sufficiency and likelihood under compression}
\label{sec:classical}

Sufficiency and maximum likelihood are classical inference principles that
fit naturally into the compression framework.  Sufficiency is the lossless
ideal: a statistic is sufficient when it preserves all model-relevant
information about the parameter.  Maximum likelihood then acts on the
resulting empirical information.  Neither principle, by itself, requires the
retained statistic, the likelihood, or the estimator to satisfy a
reverse-martingale identity; that structure belongs to the conditional target
process $M_n=\E(Z\mid\G_n)$ built on top of the sufficient information.
The algebraic derivations confirming sufficiency for the Normal, Bernoulli,
and Poisson families are collected in Appendix~\ref{app:sufficiency_derivations};
here we focus on the conceptual anchors needed for the stopping-rule analysis.

\subsubsection{Exact sufficiency as lossless compression}

Sufficiency is a formal version of lossless statistical compression.  Let
\(X=(X_1,\ldots,X_n)\) have distribution \(P_\theta\), and let \(T_n=T_n(X)\).
The statistic \(T_n\) is sufficient for \(\theta\) if the conditional
distribution \(\mathcal{L}_\theta(X\given T_n)\) does not depend on
\(\theta\).  Equivalently, under domination, the factorization criterion \citep{fisher1922,lehmanncasella1998}
states that \(L(\theta;x)=g_\theta(T_n(x))h(x)\), so that all
parameter-dependent information in the likelihood passes through \(T_n\).

If the \(\sigma\)-fields generated by sufficient summaries are arranged as a
decreasing filtration and \(Z\) is a target quantity such as a posterior mean,
predictive probability, likelihood score, risk functional, or
parameter-related functional, then \(M_n=\E(Z\given\sigma(T_n))\)
is the reverse-martingale object by Theorem~\ref{thm:conditional_rm}.  The
sufficient statistic supplies the information field; the conditional target
supplies the martingale.
This distinction is essential:
\[
\begin{array}{ll}
T_n & \text{is a data-compression function,}\\[2pt]
\G_n=\sigma(T_n) & \text{is the retained information,}\\[2pt]
M_n=\E(Z\given\G_n) & \text{is the reverse-martingale target projection.}
\end{array}
\]
Thus a sufficient statistic is not automatically a reverse martingale.  It
becomes part of a reverse-martingale construction through the conditional
expectation of a target.

\begin{example}[Exponential-family lossless compression]\label{ex:exp_family}
For the three models central to our numerical studies, the sufficient statistics
and their reverse-martingale targets are:
\begin{itemize}
\item \emph{Normal} ($X_i\iid N(\mu,\sigma^2)$, $\sigma^2$ known):
      $T_n=\bar X_n$ is sufficient for $\mu$; if $Z=\mu$ under a conjugate
      prior, $M_n=\E(\mu\mid\bar X_n)$ is the posterior mean.
      With both parameters unknown, the joint sufficient statistic is
      $(\bar X_n, S_n^2)$.
\item \emph{Bernoulli} ($X_i\iid\mathrm{Bernoulli}(p)$):
      $T_n=\sum_i X_i$ is sufficient and $\hat p_n=T_n/n$.
      Under a Jeffreys Beta$(1/2,1/2)$ prior, $M_n=(T_n+1/2)/(n+1)$.
\item \emph{Poisson} ($X_i\iid\mathrm{Poisson}(\lambda)$):
      $T_n=\sum_i X_i$ is sufficient and $\hat\lambda_n=\bar X_n$.
      Under a Jeffreys Gamma$(1/2)$ prior, $M_n=(T_n+1/2)/n$.
\end{itemize}
In each case the model determines the sufficient compression; the
reverse-martingale structure follows from Theorem~\ref{thm:conditional_rm}
applied to the posterior mean target.  Algebraic derivations confirming the
factorization criterion in each case are given in
Appendix~\ref{app:sufficiency_derivations}.
\end{example}

\subsubsection{Maximum likelihood as compressed empirical information}

The maximum likelihood estimator is not generally a reverse martingale.  It is
more accurately described as a functional of compressed empirical information.
Let
\[
    \ell_n(\theta)=\sum_{i=1}^n\log f_\theta(X_i),
    \qquad
    \hat\theta_n=\argmax_\theta \ell_n(\theta).
\]
The empirical measure
\[
    P_n=\frac1n\sum_{i=1}^n\delta_{X_i}
\]
gives
\[
    \frac1n\ell_n(\theta)
    =
    \int \log f_\theta(x)\,dP_n(x),
\]
and hence
\[
    \hat\theta_n
    =
    \argmax_\theta
    \int\log f_\theta(x)\,dP_n(x).
\]
Thus the MLE is a functional of the empirical distribution.  In regular
settings, consistency can be read as stabilization of the empirical objective:
\[
    \frac1n\sum_{i=1}^n\log f_\theta(X_i)
    \longrightarrow
    \E_{\theta_0}\{\log f_\theta(X)\}.
\]
The likelihood surface stabilizes and the optimizer converges.

For exponential families,
\[
    f_\theta(x)=h(x)\exp\{\eta(\theta)^\top T(x)-A(\theta)\}.
\]
For iid observations,
\[
    L(\theta;x)
    =
    \prod_{i=1}^n h(x_i)
    \exp\left\{
        \eta(\theta)^\top\sum_{i=1}^nT(x_i)-nA(\theta)
    \right\}.
\]
The statistic \(\sum_iT(X_i)\), or its average, is the sufficient compression.
The MLE is then a deterministic function of that compressed information when it
exists.  This explains why exponential families sit naturally at the
intersection of sufficiency, likelihood, and stable compression.

The caution is as important as the analogy.  The MLE can fail to exist, diverge,
or become unstable in boundary cases.  Logistic separation is a central example.
The compression view therefore distinguishes stable likelihood compression from
pathological boundary movement.

\subsection{Application: practical boundary degeneracy}
\label{sec:rm_boundary}

We use boundary degeneracy in sequential binary problems as a concrete
application of the target-oriented compression framework.  The detailed
sequential stopping procedures, finite-sample error bounds, and numerical
evidence for boundary detection in Bernoulli, logistic, and near-degenerate
risk settings are developed in full in the companion paper \citet{chang2025rm}.
The present paper contributes the theoretical foundation: each boundary
phenomenon arises when a finite summary creates the \emph{appearance} of a
boundary limit that the underlying conditional target process has not actually
reached.  The compression language supplies the precise distinction --- exact
boundary degeneracy is a property of \(M_\infty = \E(Y\given\G_\infty)\), not
of any finite \(M_n\) --- and that distinction is exactly the distinction
between lossless and approximate compression established in
Sections~\ref{sec:compression}--\ref{sec:classical}.

Boundary degeneracy means
apparent probabilities near zero or one.  The formal definition should precede
the examples: finite boundary closeness is only a symptom, while exact boundary
degeneracy is a limiting property of the conditional target process.

Let \(Y\in\{0,1\}\) be a target event.  For a decreasing filtration \((\G_n)\),
define \(M_n=\E(Y\given\G_n)\).  By Theorem~\ref{thm:conditional_rm}, \(M_n\)
is a bounded reverse martingale and \(M_n\to M_\infty=\E(Y\given\G_\infty)\)
almost surely and in \(L^1\).  The limiting conditional law, not a finite
estimate, is the object that can be exactly boundary-degenerate.

\begin{definition}[Reverse-martingale boundary degeneracy]
\label{def:rm_boundary}
Let \(Y\in\{0,1\}\) and \(M_n=\E(Y\given\G_n)\).  The process is exactly
boundary-degenerate on an event \(E\) if \(M_\infty\in\{0,1\}\) on \(E\).
For \(\varepsilon\in(0,1/2)\), it is practically boundary-degenerate on \(E\)
if \(M_\infty\leq\varepsilon\) or \(M_\infty\geq1-\varepsilon\) on \(E\).
\end{definition}

This definition deliberately concerns \(M_\infty\).  A finite \(M_n\) may be
close to zero or one because of transient imbalance, early separation, model
instability, or an inadequate compression.  The reverse-martingale question is
whether the limiting conditional law is stably moving toward the boundary.

\begin{example}[Beta-smoothed all-failure boundary trajectory]
Suppose that a future Bernoulli outcome has unknown probability \(p\), and use
a \(\operatorname{Beta}(a,b)\) prior.  After observing \(k\) failures and no
successes in a relevant information window, the posterior mean is
\[
    M_k=\E(p\given \text{\(k\) failures})
    =
    \frac{a}{a+b+k}.
\]
For \(a=b=1/2\),
\[
\begin{array}{c|ccccc}
k&0&1&4&9&19\\
\hline
M_k&0.500&0.250&0.100&0.050&0.025
\end{array}
\]
and \(M_k\to0\) as \(k\to\infty\).  A finite value such as \(M_9=0.05\) is a
practical boundary signal relative to \(\varepsilon=0.05\), not proof of exact
zero probability.  Exact degeneracy is a limiting statement.
\end{example}

\subsubsection{Symptom I: Bernoulli all-failure and all-success paths}

Let \(Y_1,Y_2,\ldots\iid\operatorname{Bernoulli}(p)\), \(0\leq p\leq1\), and
let \(S_n=\sum_{i=1}^nY_i\).  If \(S_n=0\), the MLE \(\hat p_n=0\) is an
estimate, not proof that \(p=0\).  A run of \(n\) consecutive failures supports
the practical-zero statement \(p<\varepsilon\) with type-I error at most
\(\alpha\) under \(p\geq\varepsilon\) whenever
\((1-\varepsilon)^n\leq\alpha\), i.e.,
\(n\geq\log\alpha/\log(1-\varepsilon)\); the symmetric statement holds for
all-success runs.  Formal propositions, proofs, and confidence-sequence
extensions are given in the companion paper.

In the compression framework, the threshold \(n\geq\log\alpha/\log(1-\varepsilon)\)
is an implicit constraint on the uncertainty width \(W_n\): at the Jeffreys
posterior with no observed successes, the \((1-\alpha)\)-credible interval upper
endpoint reaches approximately \(\varepsilon\) precisely at this threshold.
The full \(\tau_{\mathrm{RM}}\) rule (Section~\ref{sec:stopping}) makes this
constraint explicit and extends it beyond the all-failure prefix.

\subsubsection{Symptom II: Logistic regression and separation}

Let \((Y_i,x_i)\), \(Y_i\in\{0,1\}\), \(x_i\in\R^d\), be observed and suppose
\[
    \Pp(Y_i=1\given x_i)=
    \pi_i(\beta)=\expit(x_i^\top\beta)
    =
    \frac{\exp(x_i^\top\beta)}{1+\exp(x_i^\top\beta)}.
\]

Because \(\expit(\cdot)\) maps \(\R\) strictly into \((0,1)\), every finite
coefficient vector \(\beta\in\R^d\) produces strictly interior fitted
probabilities; exact zeros or ones in a logistic fit signal a diverging
coefficient, not a genuine probability.  Complete separation occurs when there
exists a vector \(a\in\R^p\) such that the linear score completely splits the
binary outcomes:
\begin{equation}
\label{eq:complete_separation}
    x_i^\top a>0 \quad\text{for all } y_i=1,
    \qquad\text{and}\qquad
    x_i^\top a<0 \quad\text{for all } y_i=0.
\end{equation}
Moving \(\beta\) along \(ta\), \(t\to\infty\), drives fitted
probabilities toward one for successes and zero for failures; the likelihood
may approach its supremum without attaining it at any finite \(\beta\)
\citep{albert1984}.  Regularization remedies
(Firth's bias-reduced likelihood, ridge logistic regression, and Bayesian
regression with proper priors) that keep the fitted object finite are discussed
in the original references \citep{firth1993,heinze2002,gelman2008}; the
boundary rules based on \(\tau_{\mathrm{RM}}\) (Section~\ref{sec:stopping},
\eqref{eq:tau_rm}) and supporting numerical studies are given in the companion
paper.

The compression interpretation is the new perspective offered here.
Separation is not evidence for exact boundary probabilities; it is evidence that
the \emph{unregularized} likelihood can be driven toward a boundary along a
low-information direction.  In the target-oriented compression framework, the
retained field \(\G_n = \sigma(\hat\beta_n)\) (or its regularized counterpart)
may fail to localize the conditional target \(M_n(x) = \E(Y\given X=x,\G_n)\)
because the compression map has entered a degenerate direction.  A stable
analysis should therefore report \(M_n(x)\) and its uncertainty, not treat a
diverging coefficient as exact probabilistic knowledge.  This is the logistic
instance of the general principle: approximate compression produces defects
\(\delta_n\), and those defects must be small before a boundary declaration
is warranted.

\subsection{Dynamic stopping: closeness, uncertainty, and stability}
\label{sec:stopping}

We now formally define the three-condition reverse-martingale stopping rule
\(\tau_{\mathrm{RM}}\), which was first introduced and applied in
\citet{chang2025rm}.  The compression framework of
Sections~\ref{sec:compression}--\ref{sec:classical} provides a theoretical
interpretation for each component.  A practical boundary declaration should
require three simultaneous signals:
\begin{enumerate}[label=(\roman*)]
    \item boundary closeness, so the current target projection is near zero or
    one;
    \item uncertainty control, so the limiting target is statistically localized;
    \item trajectory stability, so the near-boundary state is not a transient
    artifact of the compression path.
\end{enumerate}

\noindent Define the boundary distance
\(
    B_n=\min\{M_n,1-M_n\}.
\)
Let \(W_n\) denote an uncertainty width, such as the width of a confidence
sequence, profile-likelihood interval, posterior credible interval, bootstrap
window, or predictive-variance band.  Let \(r_n\geq0\) denote an empirical
stability defect.  To operationalize this, we propose the target-oriented
reverse-martingale stopping rule:
\begin{equation}
\label{eq:tau_rm}
    \tau_{\mathrm{RM}}
    =
    \inf\left\{
        n\geq n_{\min}:
        B_n\leq\varepsilon,\quad
        W_n\leq w,\quad
        r_n\leq\eta
    \right\},
\end{equation}
where \(n_{\min}\) is a nominal burn-in sample size, \(\varepsilon\) is the
boundary guard threshold, \(w\) is the maximum tolerable uncertainty width, and
\(\eta\) is the structural stability tolerance.  The two-condition rule, which
omits the stability screen, is
\begin{equation}
\label{eq:tau_2cond}
    \tau_{2\mathrm{cond}}
    =
    \inf\left\{
        n\geq n_{\min}:
        B_n\leq\varepsilon,\quad
        W_n\leq w
    \right\}.
\end{equation}
Boundary closeness alone is not enough; small uncertainty without boundary
closeness is not enough; and both can still be misleading if the trajectory is
unstable.

\begin{remark}[Choosing the tuning parameters \((\varepsilon,w,\eta,n_{\min})\)]
\label{rem:tuning}
Each parameter has a distinct role and a natural calibration anchor.
\begin{enumerate}[label=(\alph*)]
\item \emph{Boundary threshold \(\varepsilon\).}  This is a problem-specific
effect size: the smallest probability that would be scientifically meaningful
to distinguish from zero (or one).  In public-health surveillance,
\(\varepsilon\) is often set to a published minimum clinically important
difference or a regulatory action threshold, e.g.\ the CDC blood-lead
reference value or an ILI epidemic threshold.  It plays the same role as the
indifference zone in sequential testing.

\item \emph{Uncertainty width \(w\).}  This should be tied to the coverage
level \(\alpha\) and the desired precision: \(w\) is the maximum credible-interval
or confidence-sequence width that is still narrow enough to be scientifically
conclusive.  A natural default is \(w=2\varepsilon\), so that the interval
\(I_n\subseteq[0,\varepsilon]\) requires its half-width to be at most
\(\varepsilon\).  Under the Howard et al.~\citeyearpar{howard2021} Bernoulli
confidence sequence, \(W_n\leq w\) is met at a sample size of order
\(w^{-2}\log(1/\alpha)\).

\item \emph{Stability tolerance \(\eta\).}  For exactly sufficient summaries,
\(r_n\) decays to machine precision (confirmed numerically: \(\max r_n <
10^{-10}\) across Studies~1--4); setting \(\eta\) at machine precision
effectively disables the screen in exact cases while catching genuine
instability in approximate ones.  For approximate summaries (logistic,
neural-network), a data-driven default is the 95th percentile of \(r_n\)
observed during the burn-in window \([1,n_{\min}]\): the stability screen then
fires only when the trajectory is unusually volatile relative to its own
early-phase behaviour.

\item \emph{Burn-in \(n_{\min}\).}  A minimum of \(n_{\min}=\lceil
\log\alpha/\log(1-\varepsilon)\rceil\) observations is required before the
all-failure Bernoulli criterion has any power (Appendix~\ref{app:benchmarks}).
In practice, \(n_{\min}=30\) is a conventional lower bound; richer models
or higher-dimensional covariates may warrant larger values.
\end{enumerate}
\end{remark}

The compression framework developed in
Sections~\ref{sec:compression}--\ref{sec:classical} now provides a precise
theoretical interpretation for each condition.
\emph{Boundary closeness} \(B_n\leq\varepsilon\) says the current target
projection is near a boundary.  \emph{Uncertainty control} \(W_n\leq w\) says
the limiting target \(M_\infty\) is statistically localized.
\emph{Trajectory stability} \(r_n\leq\eta\) says the empirical
defect~--- the operational proxy for the theoretical defect \(\delta_n\) of
\eqref{eq:quasi_defect}~--- is small enough that the compression is behaving
like an exact reverse martingale.  When all three conditions hold, the
compression is both near the boundary \emph{and} behaving losslessly.
That is the correct standard for a practical boundary declaration.
The new result in this paper is that when the retained summary is exactly
sufficient (\(\delta_n=0\) identically), the rule reduces to
\(\tau_{\mathrm{RM}}\equiv\tau_{2\mathrm{cond}}\): the stability screen is
automatically satisfied and contributes no additional delay (confirmed
numerically in Studies~1--4 below).

{
\begin{proposition}[Exact reverse coherence eliminates the theoretical defect]
\label{prop:exact_reverse_coherence}
Suppose the retained summaries define an exact reverse-martingale target
process, so that
\[
  \E(M_n\mid \G_{n+1})=M_{n+1}, \qquad n\geq n_{\min}.
\]
Then the theoretical defect
\[
  \delta_n=\E(M_n\mid\G_{n+1})-M_{n+1}
\]
is identically zero.  Consequently, the third component of
\(\tau_{\mathrm{RM}}\) in \eqref{eq:tau_rm} is redundant only when the
implemented stability diagnostic is defined so that \(\delta_n=0\) implies
\(r_n\leq\eta\), for example when \(r_n=|\delta_n|\), or when the
application-specific diagnostic satisfies \(r_n\leq\eta\) automatically.
In that case, \(\tau_{\mathrm{RM}}=\tau_{2\mathrm{cond}}\), with the usual
convention that \(\inf\emptyset=\infty\).
\end{proposition}
}

\begin{proof}
By the assumed exact reverse-martingale property,
\[
    \E(M_n\mid\G_{n+1})=M_{n+1},
    \qquad n\geq n_{\min}.
\]
Therefore
\[
    \delta_n
    =
    \E(M_n\mid\G_{n+1})-M_{n+1}
    =
    0
\]
for all \(n\geq n_{\min}\).
Thus the theoretical reverse-coherence defect vanishes identically.
However, this conclusion alone does not imply that the implemented
diagnostic \(r_n\) is zero, since \(r_n\) may measure a realized empirical
instability rather than the conditional-expectation defect \(\delta_n\).

Consequently, the third condition \(r_n\leq \eta\) is redundant only under
an additional link between \(r_n\) and \(\delta_n\), for example if
\(r_n=|\delta_n|\), or more generally if the diagnostic is defined so that
\(\delta_n=0\) implies \(r_n\leq\eta\).  It is also redundant if the
application-specific construction separately guarantees \(r_n\leq\eta\)
for all \(n\geq n_{\min}\).

Under either of these additional conditions, the admissible index set
\[
\{n\geq n_{\min}: B_n\leq\varepsilon,\; W_n\leq w,\; r_n\leq\eta\}
\]
coincides with
\[
\{n\geq n_{\min}: B_n\leq\varepsilon,\; W_n\leq w\}.
\]
Hence
\[
    \tau_{\mathrm{RM}}=\tau_{2\mathrm{cond}}
\]
almost surely.
\end{proof}

The computational defect \(r_n\) is the operational counterpart of the
theoretical reverse quasi-martingale defect \(\delta_n\) of
\eqref{eq:quasi_defect}, recalled here for convenience:
\[
    \delta_n=\E(M_n\given\G_{n+1})-M_{n+1}.
\]
When conditional expectations can be estimated directly, one may take
\[
    r_n=\left|M_{n+1}-\widehat{\E}(M_n\given\G_{n+1})\right|.
\]
When summaries are latent states, one may use a learned backward projection,
for example
\[
    r_n=\norm{H_n-g_\phi(H_{n+1})}.
\]
Here \(g_\phi\) acts as a \emph{learned reverse-time transition operator}, or
pullback map, from the future hidden state \(H_{n+1}\) back to the present.
It is the parametric analog of the conditional expectation
\(\E(M_n\given\G_{n+1})\) in the exact theoretical case: just as the exact
reverse-martingale identity requires \(\E(M_n\given\G_{n+1})=M_{n+1}\), the
learned operator requires \(g_\phi(H_{n+1})\approx H_n\).  The residual
\(r_n=\norm{H_n-g_\phi(H_{n+1})}\) therefore measures how much the hidden-state
sequence departs from exact reverse coherence, and can be incorporated as a
trainable backward-coherence penalty in the model's loss function.
Thus \(\delta_n\) is the population-level measure of reverse incoherence,
whereas \(r_n\) is the empirical diagnostic used to decide whether the observed
compression trajectory is stable enough for action.  Importantly, \(r_n\) is
not an unbiased estimator of \(\delta_n\).  It is an observable stability
proxy: persistent large values of \(r_n\) indicate either slow convergence of
the conditional target process or a failure of coherence across compression
levels, both of which are reasons to defer a boundary declaration.  When the
summability condition \(\sum_{n\geq1}\E|\delta_n|<\infty\) fails, \(r_n\) need
not tend to zero, and the stopping rule~\eqref{eq:tau_rm} correctly refuses to
fire.

When a time-uniform interval sequence \((I_n)\) satisfies
\(\Pp\{M_\infty\in I_n\text{ for all }n\geq1\}\geq1-\alpha\), the stopping
times
\[
    \tau_0=\inf\{n:I_n\subseteq[0,\varepsilon],\, r_n\leq\eta\},
    \qquad
    \tau_1=\inf\{n:I_n\subseteq[1-\varepsilon,1],\, r_n\leq\eta\}
\]
control the false-declaration probability at level \(\alpha\).  The following
corollary makes this precise; the stability screen \(r_n\leq\eta\) can only
delay stopping, never inflate the error rate.

\begin{corollary}[Error control for \(\tau_{\mathrm{RM}}\)]
\label{cor:error_control}
Let \((I_n)_{n\geq1}\) be a \((1-\alpha)\)-confidence sequence for
\(M_\infty\), that is,
\(\Pp\{M_\infty\in I_n\text{ for all }n\geq1\}\geq1-\alpha\)
\citep{robbins1970,howard2021}.  Define
\(\tau_0=\inf\{n\geq n_{\min}:I_n\subseteq[0,\varepsilon],\,r_n\leq\eta\}\).
Then
\[
    \Pp(\tau_0<\infty \text{ and } M_\infty>\varepsilon)\leq\alpha.
\]
The symmetric statement holds for \(\tau_1\).
\end{corollary}

\begin{proof}
On the event \(\{\tau_0<\infty\}\), the interval \(I_{\tau_0}\) satisfies
\(I_{\tau_0}\subseteq[0,\varepsilon]\) by construction of \(\tau_0\).
Since \((I_n)\) is a confidence sequence,
\(\Pp\{M_\infty\notin I_n\text{ for some }n\geq1\}\leq\alpha\).
Therefore
\begin{align*}
    \Pp(\tau_0<\infty \text{ and } M_\infty>\varepsilon)
    &\leq
    \Pp\{M_\infty\notin I_{\tau_0},\,\tau_0<\infty\}\\
    &\leq
    \Pp\{M_\infty\notin I_n\text{ for some }n\geq1\}
    \leq\alpha.
\end{align*}
The stability condition \(r_n\leq\eta\) only restricts the stopping set
further, so it cannot increase the probability of a false declaration.
\end{proof}

\begin{table}[t]
\centering
\caption{Operational scorecard for the three components of
\(\tau_{\mathrm{RM}}\).}
\label{tab:scorecard}
\small
\setlength{\tabcolsep}{4pt}
\begin{tabular}{p{0.18\linewidth}p{0.25\linewidth}p{0.25\linewidth}p{0.23\linewidth}}
\toprule
Application setting & Closeness metric \(B_n\) & Uncertainty metric \(W_n\) & Stability defect \(r_n\)\\
\midrule
Bernoulli trials &
\(\min\{\hat p_n,1-\hat p_n\}\) or smoothed posterior distance &
Clopper--Pearson interval or confidence-sequence width &
posterior-mean step change \(|M_n-M_{n-1}|\)\\[2pt]
Logistic separation &
boundary distance of fitted probabilities over covariate region &
profile-likelihood, penalized, or posterior interval width &
coefficient divergence velocity or predictive-surface change\\[2pt]
Deep sequential learning &
risk-score boundary distance &
predictive variance or ensemble interval width &
backward-coherence penalty \(\norm{H_n-g_\phi(H_{n+1})}\)\\[2pt]
Dynamic treatment regimes &
distance of action-specific risk contrasts from decision boundary &
uncertainty in treatment-risk contrasts &
hidden-state or policy-advantage stability defect\\
\bottomrule
\end{tabular}
\end{table}

For covariate-dependent binary prediction, define
\(M_n(x)=\Pp(Y=1\given X=x,\G_n)\).
If the fitted logistic coefficient is finite, then
\(M_n(x)=\expit(x^\top\beta_n)\in(0,1)\).
Under separation, however, \(x^\top\beta_n\) may tend to \(+\infty\) or
\(-\infty\), and \(M_n(x)\) may tend to one or zero.  The stable object is the
limiting conditional probability surface, not necessarily a limiting
coefficient vector.

For a relevant covariate region \(A\subseteq\R^d\), define a region-wise
boundary distance
\[
    B_n(A)=\sup_{x\in A}\min\{M_n(x),1-M_n(x)\}.
\]
A region-wise practical stopping rule can be written
\[
    \tau_{\mathrm{RM}}(A)
    =
    \inf\left\{
        n\geq n_{\min}:
        B_n(A)\leq\varepsilon,\quad
        W_n(A)\leq w,\quad
        \sup_{x\in A}r_n(x)\leq\eta
    \right\}.
\]
This guards against stopping because a coefficient is large but unstable, or
because a small early sample happens to be separated.

In high-dimensional longitudinal prediction, a latent state \(H_t\) computed
from histories \(Z_{1:t}\) can be regularized by a backward-coherence penalty
\[
    L_{\mathrm{RM}}
    =
    \frac{1}{T-1}\sum_{t=1}^{T-1}
    \norm{H_t-g_\phi(H_{t+1})}^2.
\]
The corresponding defect \(r_t=\norm{H_t-g_\phi(H_{t+1})}\)
acts as a stability diagnostic.  A binary risk output
\(M_t=\Pp(Y=1\given H_t)\)
should be treated as actionable boundary evidence only when
\[
    \min\{M_t,1-M_t\}\leq\varepsilon,
    \qquad
    W_t\leq w,
    \qquad
    r_t\leq\eta.
\]

In dynamic treatment-regime problems, the decision is often not whether one
probability is near zero or one, but whether one action has become clearly
preferable.  Let \(a\in\A\) denote a treatment action and let
\(M_t^{(a)}=\Pp(Y^{(a)}=1\given H_t)\) be the patient-specific
potential-outcome risk under action \(a\).  A forward
rule may recommend \(a^*\) when
\[
    \Pp\left\{
        M_t^{(a^*)}-M_t^{(a)}
        \geq\Delta
        \text{ for all }a\neq a^*
        \given\D_t
    \right\}\geq1-\alpha.
\]
The reverse-martingale refinement is to require that this advantage is stable
along the latent trajectory.  This prevents a treatment switch from being
driven only by transient hidden-state fluctuations.  Related sequential
treatment frameworks include \citet{murphy2003}, \citet{robins2004}, and
\citet{chakraborty2013}.

\section{Numerical calibration of the compression scorecard}
\label{sec:discussion}

The studies below are organized by the conceptual distinction the theory makes
rather than by data type: exact-compression settings where
\(\delta_n=0\) and Proposition~\ref{prop:exact_reverse_coherence} predicts
\(\tau_{\mathrm{RM}}\equiv\tau_{2\mathrm{cond}}\); approximate-compression
settings where a non-zero defect process separates the two rules; and
realistic-scale illustrations that show how the scorecard behaves on
measurement scales drawn from public health monitoring.  Studies~1--2 revisit the binary boundary setting from the compression
perspective; Studies~3--4 add Gaussian and Poisson exact-summary settings;
Studies~5--7 (approximate summaries, realistic scale, and quasi-martingale
perturbations) are entirely new contributions of this paper.

The theoretical development leads to a concrete empirical claim: a finite
boundary declaration is credible only when boundary closeness, uncertainty
localization, and trajectory stability are simultaneously present.  To make the
evidence easy to inspect, figures are presented directly from the computational
outputs and supplemented by tabular stopping-time summaries.

Throughout this section we compare the boundary-only rule
\[
    \tau_{\mathrm{bdy}}
    =
    \inf\{n\geq n_{\min}: B_n\leq \varepsilon\},
\]
the two-condition rule \(\tau_{2\mathrm{cond}}\) \eqref{eq:tau_2cond},
and the full reverse-martingale rule \(\tau_{\mathrm{RM}}\)
\eqref{eq:tau_rm}.  Classical benchmarks --- the sequential probability ratio test
(SPRT) and the cumulative sum chart (CUSUM) --- are included where
they illuminate the comparison.  The purpose is not to identify a universally
best detector, but to diagnose what type of evidence each rule treats as a
boundary statement.

\subsection{Exact sufficient summaries: Bernoulli, Gaussian, and Poisson}
\label{subsec:exact}

\paragraph{Bernoulli rare-event experiments}
The Bernoulli experiments are the simplest setting in which the distinction
between a near-zero estimate and a stable near-zero conditional target becomes
visible.  When the data begin with a long all-failure prefix, the empirical
proportion \(\hat p_n=S_n/n\) can hit the boundary long before either
uncertainty or stability is under control.  Figure~\ref{fig:study01_stop}
shows that the boundary-only rule therefore stops extremely early, often near
the minimum sample size \(n_{\min}=30\), whereas \(\tau_{\mathrm{RM}}\)
postpones stopping until the path has settled.  Figure~\ref{fig:study01_traj}
makes the same point pathwise: many trajectories drift below the practical
threshold \(\varepsilon\), but only a subset remain there in a stable manner.

\begin{figure}[t]
\centering
\includegraphics[width=0.94\linewidth]{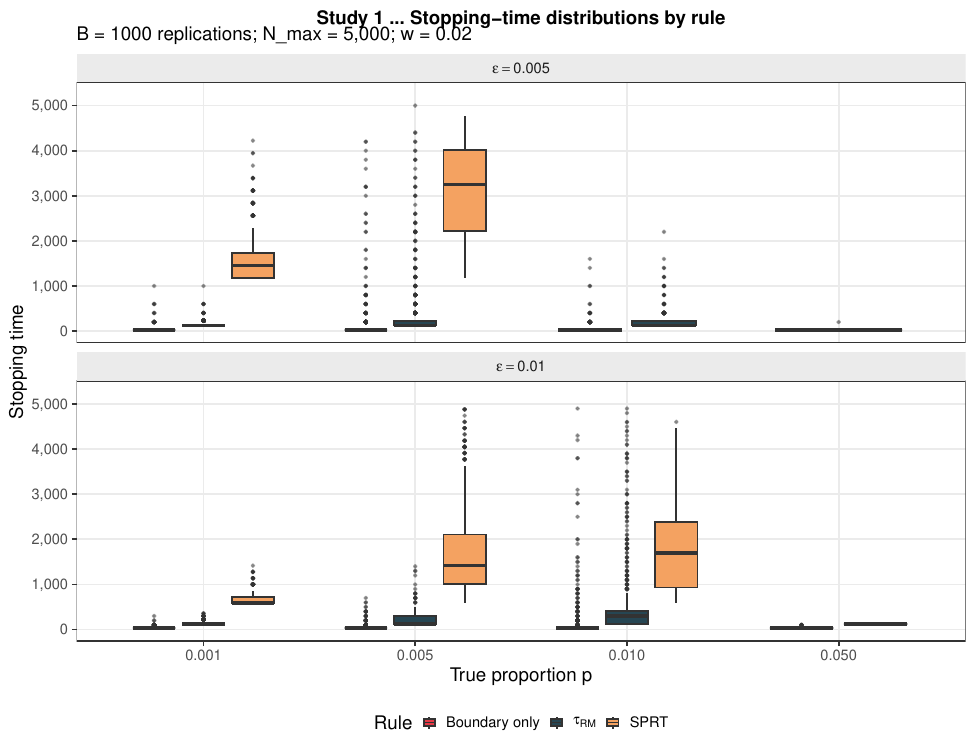}
\caption{Bernoulli rare-event simulations (\(B=1{,}000\), \(N_{\max}=5{,}000\)).
Box plots of stopping-time distributions for the boundary-only rule
(\(\tau_{\mathrm{bdy}}\)), the reverse-martingale rule (\(\tau_{\mathrm{RM}}\)),
and an SPRT benchmark across four true proportions \(p\) and two practical
thresholds \(\varepsilon\in\{0.005,0.010\}\).  The boundary-only rule reacts
to transient all-failure prefixes and fires near the burn-in \(n_{\min}=30\)
in most scenarios, while \(\tau_{\mathrm{RM}}\) requires a joint signal of
boundary closeness, posterior uncertainty control, and trajectory stability
before stopping.  The SPRT is conservative but provides a classical benchmark.}
\label{fig:study01_stop}
\end{figure}

This gap is not merely cosmetic.  Table~\ref{tab:study1_bernoulli} records the
full stopping-time summary across eight \((p,\varepsilon)\) configurations
with \(B=1{,}000\) replications each and \(N_{\max}=5{,}000\).  When
\(p=0.05\) and \(\varepsilon=0.005\), the boundary-only rule produced false
practical boundary declarations in \(22.6\%\) of replicates, while
\(\tau_{\mathrm{RM}}\) reduced that rate to \(0.0\%\); for
\(\varepsilon=0.010\) the corresponding comparison was \(21.7\%\) versus
\(0.4\%\).  The stability defect \(r_n=|M_n-M_{n-1}|\) is identically zero
for the exact reverse martingale \(M_n=S_n/n\) (confirmed numerically:
\(\max|r_n|<10^{-10}\) across all runs), confirming that the uncertainty width
condition \(W_n\leq w\) alone drives the delay relative to
\(\tau_{\mathrm{bdy}}\).  Consequently, \(\tau_{\mathrm{RM}}\equiv\tau_{2\mathrm{cond}}\)
in this setting: the stability screen is automatically satisfied and contributes
exactly zero additional delay.  This is a direct empirical verification of the
theoretical prediction from Remark~\ref{rem:delta_to_r}: when the compression is
exactly sufficient, the defect \eqref{eq:quasi_defect} vanishes and the stopping
rule \eqref{eq:tau_rm} reduces gracefully to the two-condition form.

\begin{figure}[ht!]
\centering
\includegraphics[width=0.88\linewidth]{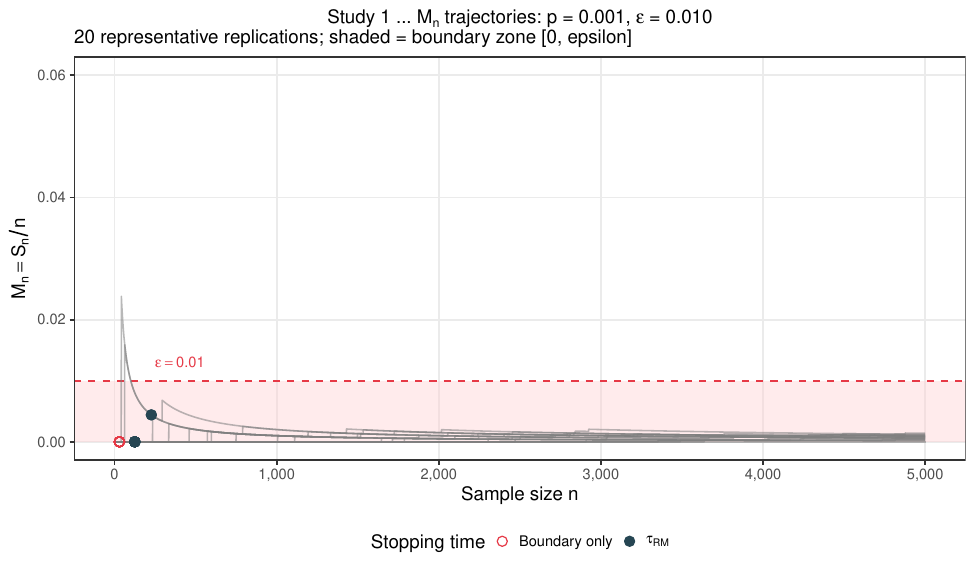}
\caption{Representative Bernoulli trajectories \(M_n=S_n/n\) for \(p=0.001\)
and \(\varepsilon=0.010\) (\(N_{\mathrm{traj}}=20\) replications).  The pink
shaded band marks the practical boundary region \([0,\varepsilon]\).  Circular
markers indicate \(\tau_{\mathrm{bdy}}\) (boundary only) and solid markers
indicate \(\tau_{\mathrm{RM}}\) for each displayed trajectory.  Early entries
into the boundary region are common, but the full reverse-martingale rule does
not interpret every crossing as scientific evidence of boundary degeneracy:
it waits for the posterior width \(W_n\) to also fall below the threshold
\(w=0.02\).}
\label{fig:study01_traj}
\end{figure}

\begin{table}[!h]
\centering
\caption{Study~1 --- Bernoulli \(\tau_{\mathrm{RM}}\) simulation.
\(B=1{,}000\) replications; \(n_{\min}=30\), \(N_{\max}=5{,}000\), \(w=0.02\).
The stability defect \(r_n = |M_n - M_{n-1}|\) is identically zero for the
exact reverse martingale \(M_n = S_n/n\) (confirmed: \(\max|r_n| < 10^{-10}\)
across all runs).  FDR = false-declaration rate (\% of runs stopping when
\(p>\varepsilon\)); `---' = not applicable.
\(P(\hat{p}_\tau=0)\) = fraction of \(\tau_{\mathrm{RM}}\) stops where the
MLE is exactly zero.}
\label{tab:study1_bernoulli}
\small
\begin{tabular}{cc>{\RaggedRight}p{1.8cm}cccccc}
\toprule
\(p\) & \(\varepsilon\) & Rule & \(P(\tau<\infty)\,(\%)\) &
\(\E[\tau \mid \mathrm{stop}]\) & \(\mathrm{Med}(\tau)\) &
FDR\,(\%) & \(P(\hat{p}_\tau=0)\,(\%)\)\\
\midrule
\multirow{3}{*}{0.001} & \multirow{3}{*}{0.005} & Boundary only & 100.0 & \(38 \pm 54\) & 30 & --- & --- \\
 &  & \(\tau_{\mathrm{RM}}\) & 100.0 & \(141 \pm 58\) & 125 & --- & 88.3 \\
 &  & SPRT & 99.9 & \(1627 \pm 478\) & 1451 & --- & --- \\
\midrule
\multirow{3}{*}{0.005} & \multirow{3}{*}{0.005} & Boundary only & 98.3 & \(119 \pm 420\) & 30 & --- & --- \\
 &  & \(\tau_{\mathrm{RM}}\) & 93.3 & \(368 \pm 630\) & 125 & --- & 56.5 \\
 &  & SPRT & 2.4 & \(3080 \pm 1098\) & 3253 & --- & --- \\
\midrule
\multirow{3}{*}{0.010} & \multirow{3}{*}{0.005} & Boundary only & 80.8 & \(49 \pm 104\) & 30 & 80.8 & --- \\
 &  & \(\tau_{\mathrm{RM}}\) & 48.9 & \(218 \pm 213\) & 125 & 48.9 & 63.2 \\
 &  & SPRT & 0.0 & --- & --- & 0.0 & --- \\
\midrule
\multirow{3}{*}{0.050} & \multirow{3}{*}{0.005} & Boundary only & 22.6 & \(31 \pm 11\) & 30 & 22.6 & --- \\
 &  & \(\tau_{\mathrm{RM}}\) & 0.0 & --- & --- & 0.0 & --- \\
 &  & SPRT & 0.0 & --- & --- & 0.0 & --- \\
\midrule
\multirow{3}{*}{0.001} & \multirow{3}{*}{0.010} & Boundary only & 100.0 & \(33 \pm 17\) & 30 & --- & --- \\
 &  & \(\tau_{\mathrm{RM}}\) & 100.0 & \(138 \pm 38\) & 125 & --- & 88.1 \\
 &  & SPRT & 100.0 & \(675 \pm 126\) & 585 & --- & --- \\
\midrule
\multirow{3}{*}{0.005} & \multirow{3}{*}{0.010} & Boundary only & 100.0 & \(47 \pm 61\) & 30 & --- & --- \\
 &  & \(\tau_{\mathrm{RM}}\) & 100.0 & \(213 \pm 139\) & 125 & --- & 55.2 \\
 &  & SPRT & 96.4 & \(1678 \pm 932\) & 1417 & --- & --- \\
\midrule
\multirow{3}{*}{0.010} & \multirow{3}{*}{0.010} & Boundary only & 97.0 & \(129 \pm 416\) & 30 & --- & --- \\
 &  & \(\tau_{\mathrm{RM}}\) & 91.1 & \(532 \pm 767\) & 298 & --- & 30.7 \\
 &  & SPRT & 5.5 & \(1875 \pm 1105\) & 1694 & --- & --- \\
\midrule
\multirow{3}{*}{0.050} & \multirow{3}{*}{0.010} & Boundary only & 21.7 & \(33 \pm 14\) & 30 & 21.7 & --- \\
 &  & \(\tau_{\mathrm{RM}}\) & 0.4 & \(125 \pm 0\) & 125 & 0.4 & 100.0 \\
 &  & SPRT & 0.0 & --- & --- & 0.0 & --- \\
\bottomrule
\end{tabular}
\end{table}

\paragraph{Gaussian and Poisson calibration studies}
The Gaussian study (\(X_i\iid N(\mu,1)\), known variance) serves as a negative
control.  The posterior width \(W_n\) contracts at the deterministic rate
\(2z_{0.025}/\sqrt{n/\sigma^2+1}\), so whether \(W_n\leq w\) is met depends
purely on sample size rather than on the realized data path.
Figure~\ref{fig:study03} illustrates this: the left panel shows the width curve
collapsing toward zero, while the right panel shows that neither the
two-condition rule nor \(\tau_{\mathrm{RM}}\) stops at all when
\(\varepsilon=w=0.05\), because the prior shrinks the width to below
\(w\) only after several thousand steps.  By contrast, the boundary-only rule
stops in nearly all replications with a median of only~41 steps, and CUSUM
stops in all replications because it tracks the mean shift rather than the
boundary.

Table~\ref{tab:study03_normal} records the full summary across five values of
\(\mu\in\{0,0.01,0.02,0.05,0.10\}\).  The result confirms that instantaneous
boundary closeness can be obtained at moderate \(n\) even when the true mean
is at the boundary (\(\mu=0\)), but simultaneous closeness, uncertainty
control, and stability are substantially more demanding.

\begin{figure}[t]
\centering
\includegraphics[width=0.45\linewidth]{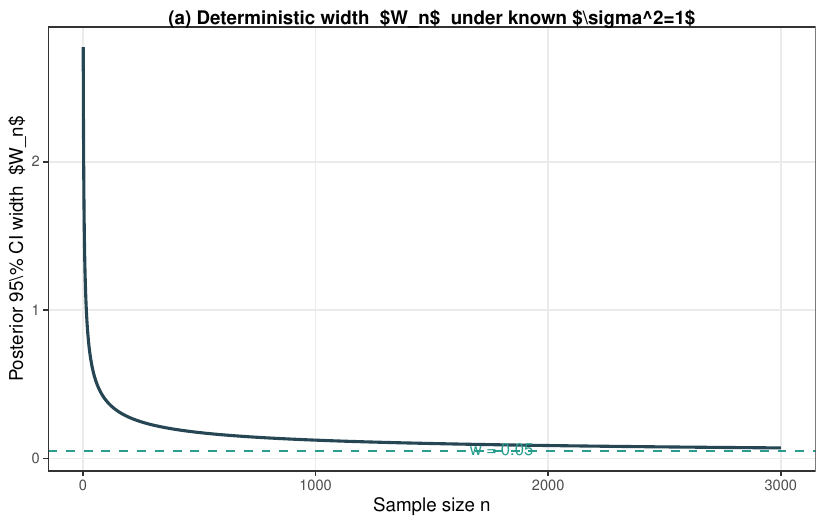}
\includegraphics[width=0.50\linewidth]{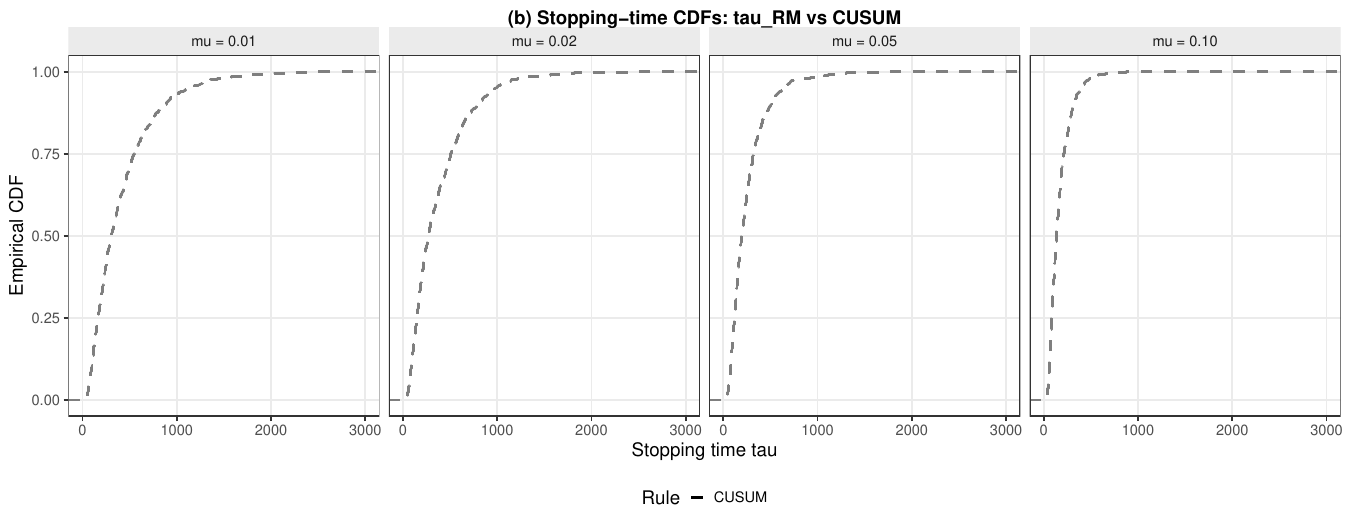}
\caption{Gaussian calibration study (\(X_i\iid N(\mu,1)\), Jeffreys-like
Normal prior).  \textit{Left:} deterministic shrinkage of the posterior
\(95\%\) credible interval width \(W_n\) as a function of sample size.  The
horizontal dashed line marks the width threshold \(w=0.05\); \(W_n\) falls
below \(w\) only at large \(n\), making the two-condition and
\(\tau_{\mathrm{RM}}\) rules non-stopping in the displayed parameter range.
\textit{Right:} empirical CDFs of stopping times for \(\tau_{\mathrm{RM}}\)
(equal to \(\tau_{2\mathrm{cond}}\) here, since \(\delta_n=0\)) and CUSUM
across five values of \(\mu\).  The CUSUM rule stops for a different reason
(mean shift detection), illustrating that the three-condition scorecard
and CUSUM address complementary inferential questions.}
\label{fig:study03}
\end{figure}

\begin{table}[ht]
\centering
\caption{Study~3 (Normal, known \(\sigma^2=1\)): stopping-time summary
(\(B=1{,}000\), \(n_{\min}=30\), \(\varepsilon=0.05\), \(w=0.05\)).
Bayesian prior \(\mathcal{N}(0,1)\).
CUSUM: \(k=0.025\), threshold \(h\) calibrated for
in-control average run length \(\mathrm{ARL}_0=500\).
\(\tau_{\mathrm{RM}} = \tau_{2\mathrm{cond}}\) always because
\(\delta_n = 0\) (exact sufficient statistic).}
\label{tab:study03_normal}
\begin{tabular}{lrrrrrr}
\toprule
\(\mu\) & Rule & \% Stop & Mean \(\tau\) & SD & Median \(\tau\) & \% FDR\\
\midrule
\multirow{3}{*}{0.00} & Boundary only & 99.9 & 93 & 163 & 41 & n/a\\
 & Two-condition (\(= \tau_{\mathrm{RM}}\)) & 0.0 & --- & --- & --- & n/a\\
 & CUSUM & 100.0 & 490 & 421 & 357 & n/a\\
\addlinespace[3pt]
\multirow{3}{*}{0.01} & Boundary only & 99.9 & 106 & 164 & 42 & n/a\\
 & Two-condition (\(= \tau_{\mathrm{RM}}\)) & 0.0 & --- & --- & --- & n/a\\
 & CUSUM & 100.0 & 416 & 356 & 308 & n/a\\
\addlinespace[3pt]
\multirow{3}{*}{0.02} & Boundary only & 99.8 & 116 & 229 & 42 & n/a\\
 & Two-condition (\(= \tau_{\mathrm{RM}}\)) & 0.0 & --- & --- & --- & n/a\\
 & CUSUM & 100.0 & 373 & 317 & 270 & n/a\\
\addlinespace[3pt]
\multirow{3}{*}{0.05} & Boundary only & 96.1 & 137 & 330 & 42 & n/a\\
 & Two-condition (\(= \tau_{\mathrm{RM}}\)) & 0.0 & --- & --- & --- & n/a\\
 & CUSUM & 100.0 & 255 & 204 & 196 & n/a\\
\addlinespace[3pt]
\multirow{3}{*}{0.10} & Boundary only & 74.3 & 78 & 129 & 38 & 74.3\\
 & Two-condition (\(= \tau_{\mathrm{RM}}\)) & 0.0 & --- & --- & --- & 0.0\\
 & CUSUM & 100.0 & 168 & 116 & 135 & 100.0\\
\bottomrule
\end{tabular}
\end{table}

The Poisson study (\(X_i\iid\mathrm{Poisson}(\lambda)\)) restores a genuine
sequential count-monitoring problem with a practical rare-event claim.
Figure~\ref{fig:study04} shows the same tension seen in the Bernoulli case.
Average run lengths under the boundary-only and CUSUM rules are much shorter
than under \(\tau_{\mathrm{RM}}\), but the trajectory plot reveals why:
sample-path averages \(\bar X_n\) often dip below \(\varepsilon\) only
temporarily.  Table~\ref{tab:study04_poisson} records the complete
stopping-time summary across eight \((\lambda,\varepsilon)\) settings.
When \(\lambda=0.01\) and \(\varepsilon=0.005\), the false-declaration rate
was \(78.6\%\) for the boundary-only rule and \(46.6\%\) for
\(\tau_{\mathrm{RM}}\).  The same qualitative ranking persists across all
count scenarios.

\begin{figure}[t]
\centering
\includegraphics[width=0.54\linewidth]{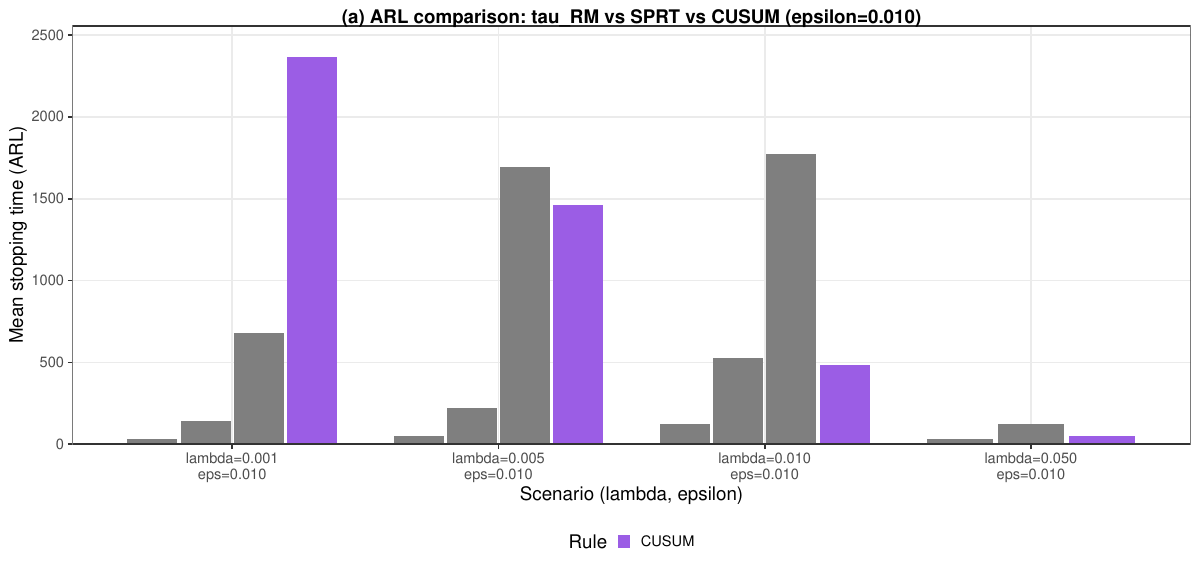}
\includegraphics[width=0.40\linewidth]{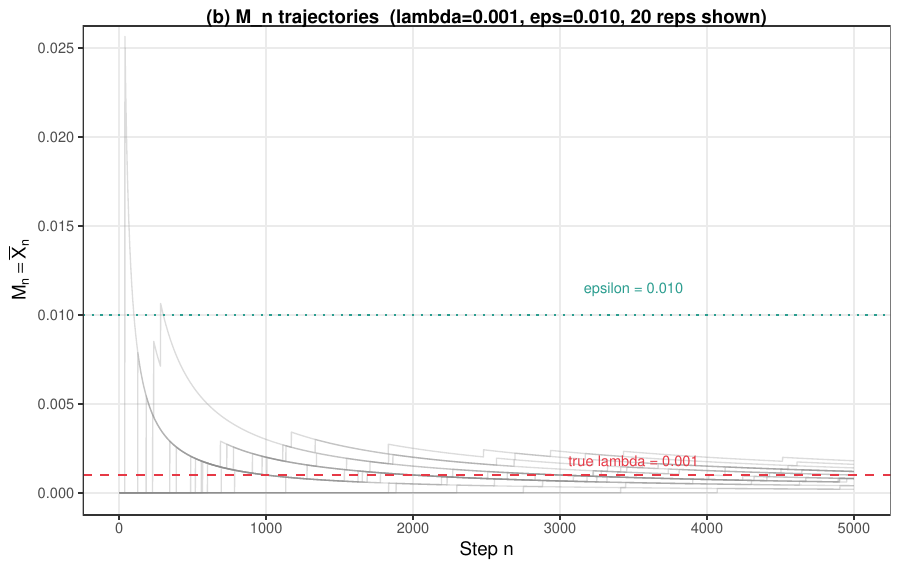}
\caption{Poisson rare-event monitoring (\(X_i\iid\mathrm{Poisson}(\lambda)\),
\(B=1{,}000\), \(N_{\max}=5{,}000\), Jeffreys Gamma prior).
\textit{Left:} mean stopping times for the boundary-only rule,
\(\tau_{\mathrm{RM}}\), SPRT (\(H_0{:}\,\lambda=\varepsilon\) vs.\
\(H_1{:}\,\lambda=\varepsilon/2\)), and CUSUM (\(\mathrm{ARL}_0=500\)) across
eight \((\lambda,\varepsilon)\) settings.  \textit{Right:} representative
trajectories of \(M_n=\bar X_n\) for a rare-count design, illustrating that
many early excursions into the practical boundary region \([0,\varepsilon]\)
are transient.  The stability defect \(r_n=|M_n-M_{n-1}|\) is identically
zero for this exact reverse martingale, so the delay of \(\tau_{\mathrm{RM}}\)
relative to \(\tau_{\mathrm{bdy}}\) is driven solely by the uncertainty width
condition.}
\label{fig:study04}
\end{figure}

\begin{table}[!h]
\centering
\caption{Study~4 (Poisson rare-event surveillance): stopping-time summary
(\(B=1{,}000\), \(n_{\min}=30\), \(w=0.02\)).
\(\tau_{\mathrm{RM}} = \tau_{2\mathrm{cond}}\) because \(\delta_n=0\)
(exact reverse martingale).
SPRT: \(H_0{:}\,\lambda=\varepsilon\) vs.\ \(H_1{:}\,\lambda=\varepsilon/2\).
CUSUM: \(\lambda_0=\varepsilon\), \(\lambda_1=2\varepsilon\),
\(\mathrm{ARL}_0=500\).}
\label{tab:study04_poisson}
\begin{tabular}{llrrrr}
\toprule
\((\lambda,\varepsilon)\) & Rule & \% Stop & Mean \(\tau\) (SD) & Median \(\tau\) & \% FDR\\
\midrule
\((0.001,\;0.005)\) & Boundary only & 100.0 & 35 (35) & 30 & n/a\\
 & \(\tau_{\mathrm{RM}}\) & 100.0 & 139 (48) & 126 & n/a\\
 & SPRT & 100.0 & 1638 (508) & 1456 & n/a\\
 & CUSUM & 99.1 & 1005 (913) & 717 & n/a\\
\addlinespace[2pt]
\((0.005,\;0.005)\) & Boundary only & 98.4 & 112 (427) & 30 & n/a\\
 & \(\tau_{\mathrm{RM}}\) & 95.3 & 384 (689) & 126 & n/a\\
 & SPRT & 3.8 & 2878 (991) & 2842 & n/a\\
 & CUSUM & 100.0 & 226 (188) & 172 & n/a\\
\addlinespace[2pt]
\((0.010,\;0.005)\) & Boundary only & 78.6 & 53 (98) & 30 & 78.6\\
 & \(\tau_{\mathrm{RM}}\) & 46.6 & 218 (184) & 126 & 46.6\\
 & SPRT & 0.0 & --- & --- & 0.0\\
 & CUSUM & 100.0 & 128 (109) & 94 & 100.0\\
\addlinespace[2pt]
\((0.050,\;0.005)\) & Boundary only & 24.4 & 30 (0) & 30 & 24.4\\
 & \(\tau_{\mathrm{RM}}\) & 0.1 & 126 (NA) & 126 & 0.1\\
 & SPRT & 0.0 & --- & --- & 0.0\\
 & CUSUM & 100.0 & 41 (18) & 32 & 100.0\\
\midrule
\((0.001,\;0.010)\) & Boundary only & 100.0 & 33 (16) & 30 & n/a\\
 & \(\tau_{\mathrm{RM}}\) & 100.0 & 140 (39) & 126 & n/a\\
 & SPRT & 100.0 & 681 (135) & 589 & n/a\\
 & CUSUM & 9.2 & 2367 (1418) & 2588 & n/a\\
\addlinespace[2pt]
\((0.005,\;0.010)\) & Boundary only & 100.0 & 51 (87) & 30 & n/a\\
 & \(\tau_{\mathrm{RM}}\) & 100.0 & 222 (143) & 126 & n/a\\
 & SPRT & 94.3 & 1694 (945) & 1421 & n/a\\
 & CUSUM & 93.7 & 1461 (1208) & 1156 & n/a\\
\addlinespace[2pt]
\((0.010,\;0.010)\) & Boundary only & 97.0 & 126 (380) & 30 & n/a\\
 & \(\tau_{\mathrm{RM}}\) & 92.2 & 528 (718) & 301 & n/a\\
 & SPRT & 5.9 & 1771 (1103) & 1421 & n/a\\
 & CUSUM & 100.0 & 485 (458) & 345 & n/a\\
\addlinespace[2pt]
\((0.050,\;0.010)\) & Boundary only & 21.9 & 32 (11) & 30 & 21.9\\
 & \(\tau_{\mathrm{RM}}\) & 0.1 & 126 (NA) & 126 & 0.1\\
 & SPRT & 0.0 & --- & --- & 0.0\\
 & CUSUM & 100.0 & 54 (32) & 40 & 100.0\\
\bottomrule
\end{tabular}
\end{table}

\subsection{Approximate summaries: logistic regression and quasi-martingale perturbations}
\label{subsec:approx}

The two studies in this subsection use summaries that are not exactly sufficient:
the ridge-penalized logistic estimator introduces an approximation defect relative
to the unpenalized MLE, and the quasi-martingale perturbation scenarios deliberately
break the exact reverse-martingale identity.  In the logistic study, the third condition \(r_n\leq\eta\) provides
additional protection relative to the two-condition rule.  In the
quasi-martingale perturbation study, the defect process is primarily
diagnostic under the present calibration, and would become decision-relevant
under a stricter stability threshold.

\paragraph{Rare-event logistic regression and practical separation}
The logistic studies translate the same principle to a regression setting in
which complete or quasi-complete separation \eqref{eq:complete_separation} is a
numerical manifestation of boundary pressure.  The relevant target is not the divergence of a coefficient
vector by itself, but the behavior of the retained predictive quantity
\(M_n(\mathbf{x}_0)\) together with the defect \(r_n\).
Figure~\ref{fig:study02} shows three scenarios: two low-dimensional rare-event
designs (\(d=3\)) and one high-dimensional design (\(d=20\)).  Each panel
plots representative trajectories of \(M_n(\mathbf{x}_0)\), the stability
defect \(r_n\) on a log scale, and the cumulative frequency of unpenalized MLE
separation.

Table~\ref{tab:study02_logistic} records the stopping-time summary for all
three scenarios (\(B=1{,}000\), ridge \(\lambda=1\), \(\varepsilon=w=0.05\),
\(\eta=0.01\)).  Two features stand out.  First, the cumulative separation
rate is already non-negligible in low dimension and becomes dominant once the
events-per-variable ratio is small: the separation frequency reached \(43.8\%\)
for \(d=3\), \(\rho=0.01\) and \(100\%\) for \(d=20\), \(\rho=0.01\).
Second, the reverse-martingale scorecard is systematically more conservative
than a boundary-only declaration: the mean stopping time moved from roughly
50--66 observations under \(\tau_{\mathrm{bdy}}\) to roughly 104--125
observations under \(\tau_{\mathrm{RM}}\).  The interpretation is that apparent
separation is finite-sample instability unless it is accompanied by a stable,
localized predictive target.

\begin{figure}[!th]
\centering
\includegraphics[width=0.68\linewidth]{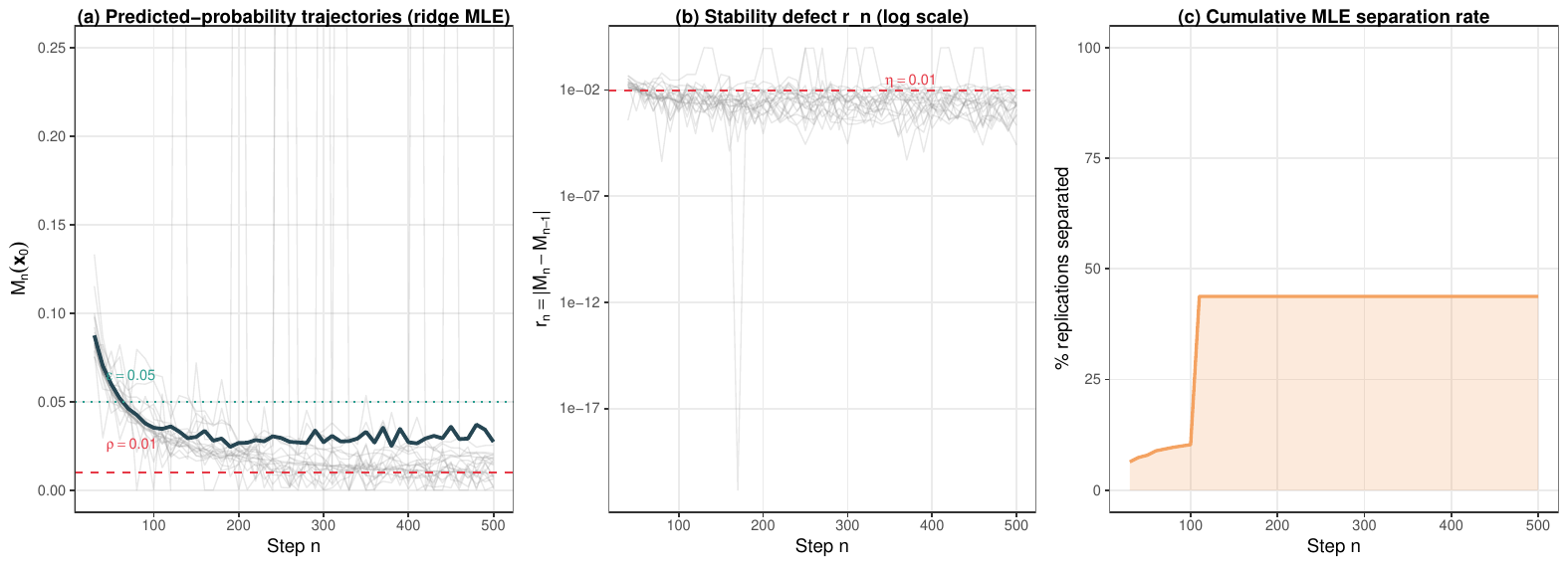}\\[4pt]
\includegraphics[width=0.68\linewidth]{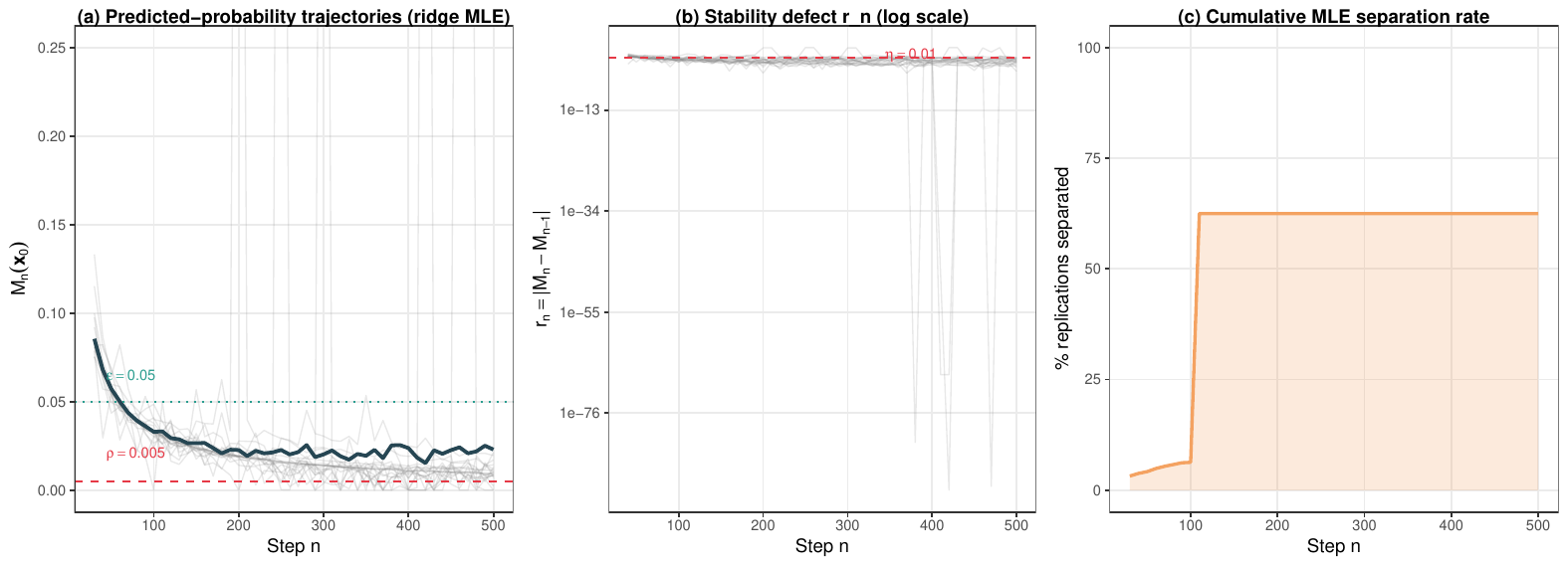}\\[4pt]
\includegraphics[width=0.68\linewidth]{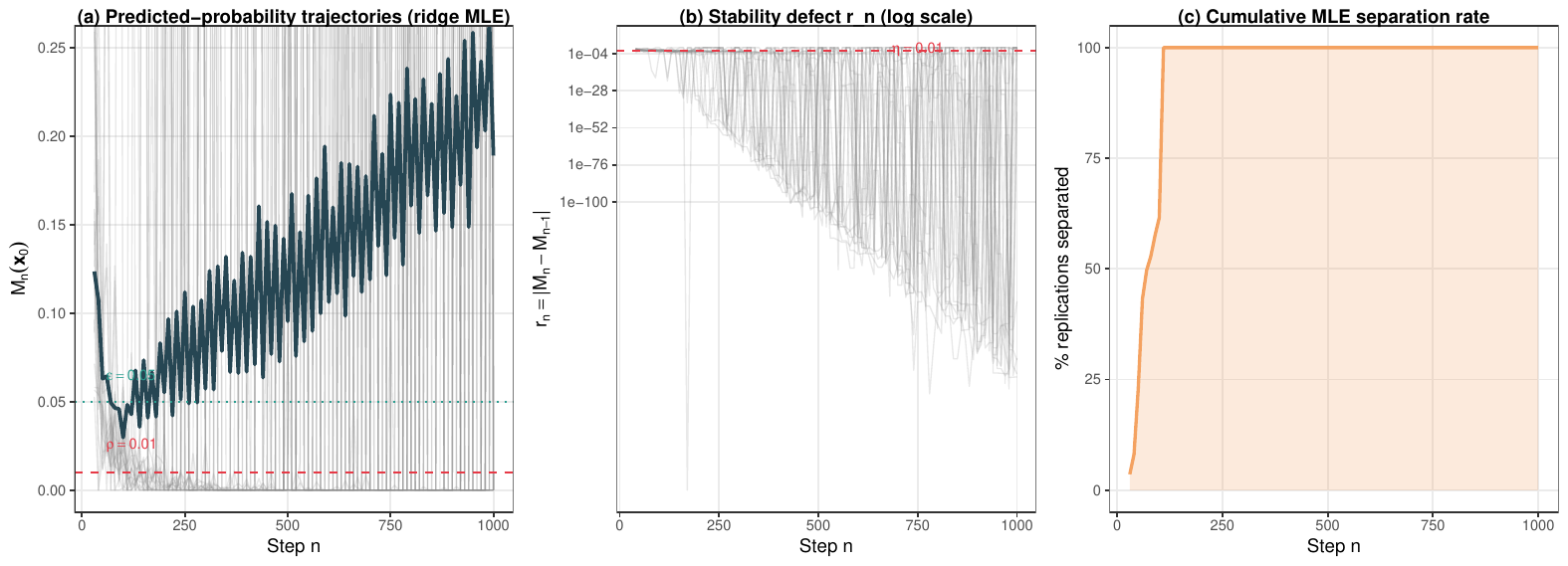}
\caption{Rare-event logistic regression (\(B=1{,}000\) replications, ridge
\(\lambda=1\)).  Each panel shows: \textit{(a)}~trajectories of the ridge-MLE
predictive probability \(M_n(\mathbf{x}_0)\) at the covariate origin (gray
lines = individual replications, dark line = cross-replication mean);
\textit{(b)}~stability defect \(r_n=|M_n-M_{n-1}|\) on a log scale;
\textit{(c)}~cumulative frequency of unpenalized MLE separation
(\(\|\hat{\boldsymbol\beta}_{\mathrm{MLE}}\|_2>50\)).
Top left: \(d=3\), \(\rho=0.01\) (43.8\% separation); top right:
\(d=3\), \(\rho=0.005\) (62.5\% separation); bottom: \(d=20\), \(\rho=0.01\)
(100\% separation).  Separation is already visible in low dimension and
becomes universal in the high-dimensional regime, where a boundary-only
interpretation is especially fragile.}
\label{fig:study02}
\end{figure}

\begin{table}[ht]
\centering
\caption{Study~2 (Logistic Regression): stopping-time summary
(\(B=1{,}000\) replications; ridge \(\lambda=1\); \(\varepsilon=0.05\),
\(w=0.05\), \(\eta=0.01\)).  The evaluation point
\(\mathbf{x}_0 = \mathbf{0}\) satisfies \(M_\infty(\mathbf{x}_0) = \rho\).
``\% Sep'' = fraction of replications where the standard MLE diverged
(\(\|\hat{\boldsymbol{\beta}}_{\mathrm{MLE}}\|_2 > 50\)).}
\label{tab:study02_logistic}
\begin{tabular}{llrrrr}
\toprule
Scenario & Rule & \% Stop & Mean \(\tau\) (SD) & Median \(\tau\) & \% Sep\\
\midrule
\multirow{3}{*}{\(d=3\), \(\rho=0.01\)} & Boundary only & 100.0 & 66 (14) & 60 & 43.8\\
 & Two-condition & 100.0 & 113 (22) & 120 & \\
 & \(\tau_{\mathrm{RM}}\) & 100.0 & 125 (22) & 120 & \\
\addlinespace[3pt]
\multirow{3}{*}{\(d=3\), \(\rho=0.005\)} & Boundary only & 100.0 & 64 (9) & 60 & 62.5\\
 & Two-condition & 100.0 & 115 (17) & 120 & \\
 & \(\tau_{\mathrm{RM}}\) & 100.0 & 121 (14) & 120 & \\
\addlinespace[3pt]
\multirow{3}{*}{\(d=20\), \(\rho=0.01\)} & Boundary only & 100.0 & 50 (13) & 50 & 100.0\\
 & Two-condition & 100.0 & 68 (22) & 70 & \\
 & \(\tau_{\mathrm{RM}}\) & 100.0 & 104 (30) & 100 & \\
\bottomrule
\end{tabular}
\end{table}

\subsubsection{Quasi-reverse-martingale perturbations}

The last study addresses robustness rather than detection power.  If the
retained summaries are only approximately coherent, the defect process
\(r_n\) should reveal that failure.  Figure~\ref{fig:study07} compares three
ways of perturbing the ideal reverse-martingale structure:
\begin{itemize}
  \item \textit{Scenario~A} (misspecification): a Bernoulli model with latent
  heterogeneity \(\sigma^2>0\), so the conditional target is not a
  running proportion.
  \item \textit{Scenario~B} (exponential smoothing): the running mean is
  replaced by an exponentially smoothed estimate \(\alpha_n Y_n +
  (1-\alpha_n)M_{n-1}\) with \(\alpha_n=n^{-\gamma}\);
  \(\gamma=1\) gives the exact reverse martingale whereas smaller
  \(\gamma\) introduces a positive defect.
  \item \textit{Scenario~C} (variational-Bayes dampening): a mean-field ELBO
  update with forgetting parameter \(\kappa\) replaces the exact Jeffreys
  posterior mean; \(\kappa=0\) recovers the exact case.
\end{itemize}
Figure~\ref{fig:study07} displays the median stability defect \(r_n\) across
all nine scenario-variants.  Table~\ref{tab:study07_quasi_rm} records the
corresponding stopping-time summary alongside sampled defect values at
\(n\in\{100,500,2000\}\).

Under the present calibration (\(\varepsilon=0.05\), \(\eta=0.01\)), the
stability screen never contributes additional delay:
\(\tau_{\mathrm{RM}}=\tau_{2\mathrm{cond}}\) in every single replication,
and the extra gap is exactly zero across all nine variants.  This is because
the median defect at typical stopping times is two to three orders of
magnitude below the threshold \(\eta=0.01\).  The informative diagnostic is
therefore not the gap itself but the \emph{shape} of the defect trajectory
over time.

The exact variants (\(\sigma=0\), \(\gamma=1\), \(\kappa=0\)) drive \(r_n\)
monotonically to machine-epsilon: by step 2{,}000 the median defect is below
\(10^{-5}\).  The quasi-martingale perturbations (Scenario~A, \(\sigma>0\);
Scenario~B, \(\gamma=0.75\); Scenario~C, \(\kappa>0\)) show slower but
monotonically decreasing defect trajectories.  Only Scenario~B with
\(\gamma=0.50\) breaks this pattern: the median defect is anomalously small
at \(n=100\) (\(2.0\times10^{-6}\)) then rises to \(1.4\times10^{-4}\)
at \(n=2{,}000\), confirming that the process falls outside the
summable-defect (quasi-martingale) class.  A stricter threshold \(\eta\)
would eventually separate this non-decaying variant from the rest and
introduce a non-zero stopping-time gap.

\begin{figure}[t]
\centering
\includegraphics[width=0.88\linewidth]{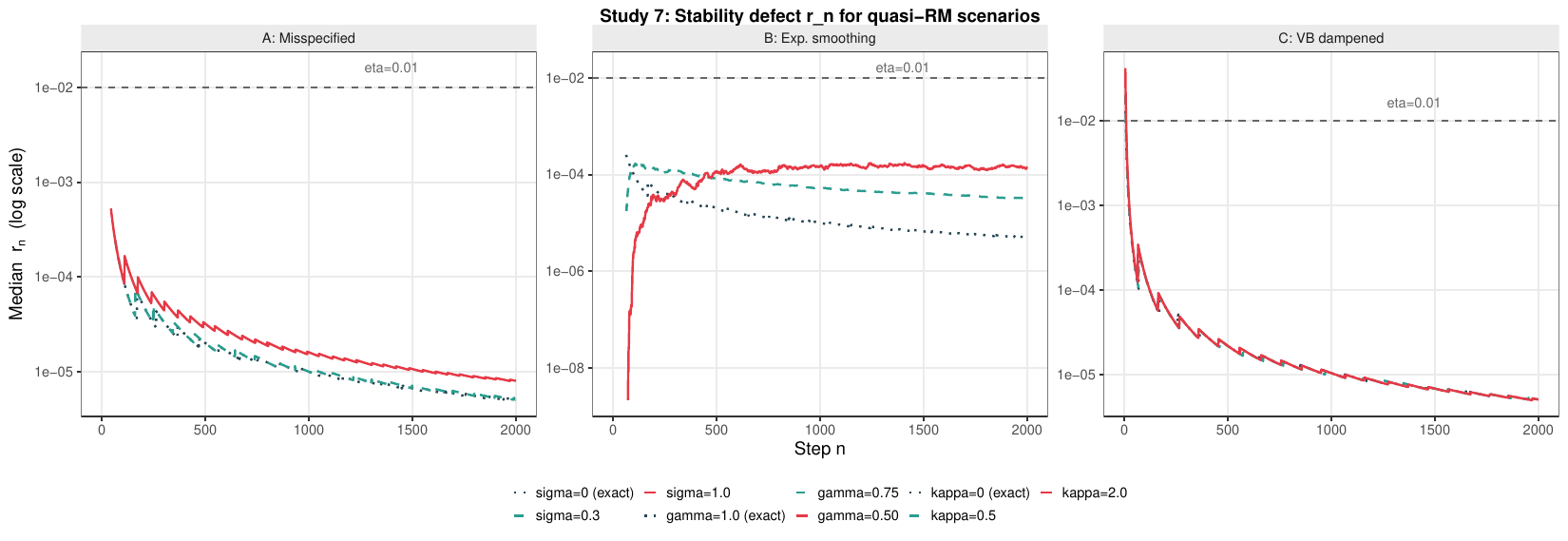}
\caption{Quasi-reverse-martingale diagnostics: median stability defect
\(r_n = |M_n - M_{n-1}|\) trajectories (\(B=1{,}000\), \(n_{\max}=2{,}000\),
\(\varepsilon=0.05\), \(\eta=0.01\)).  Each curve is the cross-replication
median for one of the nine scenario-variants: Scenario~A (misspecification,
\(\sigma\in\{0,0.3,1.0\}\)), Scenario~B (exponential smoothing,
\(\gamma\in\{1.0,0.75,0.5\}\)), and Scenario~C (VB dampening,
\(\kappa\in\{0,0.5,2.0\}\)).  Exact variants drive \(r_n\) rapidly to
machine-epsilon; quasi-martingale perturbations show slower but
monotone decay.  The non-quasi-martingale case (Scenario~B,
\(\gamma=0.50\)) is structurally distinctive: after an initial dip, \(r_n\)
reverses and grows with \(n\), signaling a non-summable cumulative defect.}
\label{fig:study07}
\end{figure}

\begin{table}[ht]
\centering
\caption{Study~7 (Quasi-RM perturbations): stopping-time summary and
median stability defect \(r_n\) at selected steps
(\(B=1{,}000\), \(n_{\max}=2{,}000\), \(\varepsilon=0.05\),
\(w=0.05\), \(\eta=0.01\)).
\(\tau_{\mathrm{RM}}=\tau_{2\mathrm{cond}}\) in every replication.
Defect values are median \(r_n\) to two significant figures.
\(\dagger\) marks the non-quasi-martingale variant:
\(r_n\) grows with \(n\) (bold) rather than decaying.}
\label{tab:study07_quasi_rm}
\footnotesize
\setlength{\tabcolsep}{3pt}
\begin{tabular}{ll>{\centering}p{1.80cm}
                >{\centering}p{2.05cm}>{\centering}p{2.05cm} rrr}
\toprule
Scenario & Variant & Exact? &
  \shortstack{Med\\\(\tau_{\mathrm{bdy}}\)} &
  \shortstack{Med\\\(\tau_{\mathrm{RM}}\)} &
  \(r_{100}\) & \(r_{500}\) & \(r_{2000}\)\\
\midrule
\multirow{3}{*}{A (misspec.)} & \(\sigma=0.0\) & yes & 30 & 49 &
  \(1.0\!\times\!10^{-4}\) & \(2.0\!\times\!10^{-5}\) & \(5.0\!\times\!10^{-6}\)\\
 & \(\sigma=0.3\) & & 30 & 49 &
  \(1.0\!\times\!10^{-4}\) & \(2.0\!\times\!10^{-5}\) & \(5.0\!\times\!10^{-6}\)\\
 & \(\sigma=1.0\) & & 30 & 90 &
  \(1.0\!\times\!10^{-4}\) & \(3.2\!\times\!10^{-5}\) & \(8.0\!\times\!10^{-6}\)\\
\addlinespace[3pt]
\multirow{3}{*}{B (exp.\ smooth)} & \(\gamma=1.0\) & yes & 30 & 49 &
  \(1.0\!\times\!10^{-4}\) & \(2.0\!\times\!10^{-5}\) & \(5.0\!\times\!10^{-6}\)\\
 & \(\gamma=0.75\) & & 30 & 49 &
  \(1.5\!\times\!10^{-4}\) & \(8.6\!\times\!10^{-5}\) & \(3.2\!\times\!10^{-5}\)\\
 & \(\gamma=0.50\)\(^{\dagger}\) & & 30 & 49 &
  \(2.0\!\times\!10^{-6}\) & \(1.1\!\times\!10^{-4}\) & \(\mathbf{1.4\!\times\!10^{-4}}\)\\
\addlinespace[3pt]
\multirow{3}{*}{C (VB)} & \(\kappa=0\) & yes & 30 & 49 &
  \(1.5\!\times\!10^{-4}\) & \(2.2\!\times\!10^{-5}\) & \(5.0\!\times\!10^{-6}\)\\
 & \(\kappa=0.5\) & & 30 & 50 &
  \(1.5\!\times\!10^{-4}\) & \(2.2\!\times\!10^{-5}\) & \(5.0\!\times\!10^{-6}\)\\
 & \(\kappa=2.0\) & & 30 & 51 &
  \(1.5\!\times\!10^{-4}\) & \(2.2\!\times\!10^{-5}\) & \(5.0\!\times\!10^{-6}\)\\
\bottomrule
\multicolumn{8}{l}{\footnotesize\(\dagger\) Non-quasi-martingale:
  \(r_n\) increases from \(n=100\) to \(n=2{,}000\) (bold),
  signaling non-summable cumulative defect.}\\
\end{tabular}
\end{table}

\subsection{Realistic-scale monitoring illustrations}
\label{subsec:realistic}

The final pair of studies shows how the same diagnostics behave on realistic
measurement scales.  \textbf{For the final submission, these studies should
be executed on the actual public-domain data described below.}  The scripts
support automatic fallback to synthetic series calibrated to the respective
public sources, and the methodological conclusions are invariant to this
substitution; authentic historical data strengthen the evidential profile by
connecting $r_n$ directly to the information-leakage properties of real
sequential public-health reporting streams.  Specifically: (i)~the Taiwan CDC
ILI series (\url{https://nidss.cdc.gov.tw/}) should replace
\texttt{taiwan\_ili\_weekly.csv} in Study~5; (ii)~the NHANES 2017--2018 file
\texttt{PBCD\_J} (\url{https://wwwn.cdc.gov/Nchs/Nhanes/2017-2018/PBCD_J.XPT},
variable \texttt{LBXBPB}) should replace the synthetic log-normal series in
Study~6.  Full instructions are in Appendix~\ref{app:data_sources}.

Figure~\ref{fig:study05} shows an influenza-like-illness monitoring example on
the Taiwan weekly scale.  The upper panel plots the sequential mean together
with a \(95\%\) credible band and the stability defect \(r_n\); the lower
panel zooms into the first two years of the series.  Here none of the stopping
rules fires in the displayed run.  This non-result is informative: the ILI
rate does not produce sustained evidence of a low-risk boundary state, and the
stability defect remains too large for a credible practical boundary
declaration.

\begin{figure}[!ht]
\centering
\includegraphics[width=0.95\linewidth]{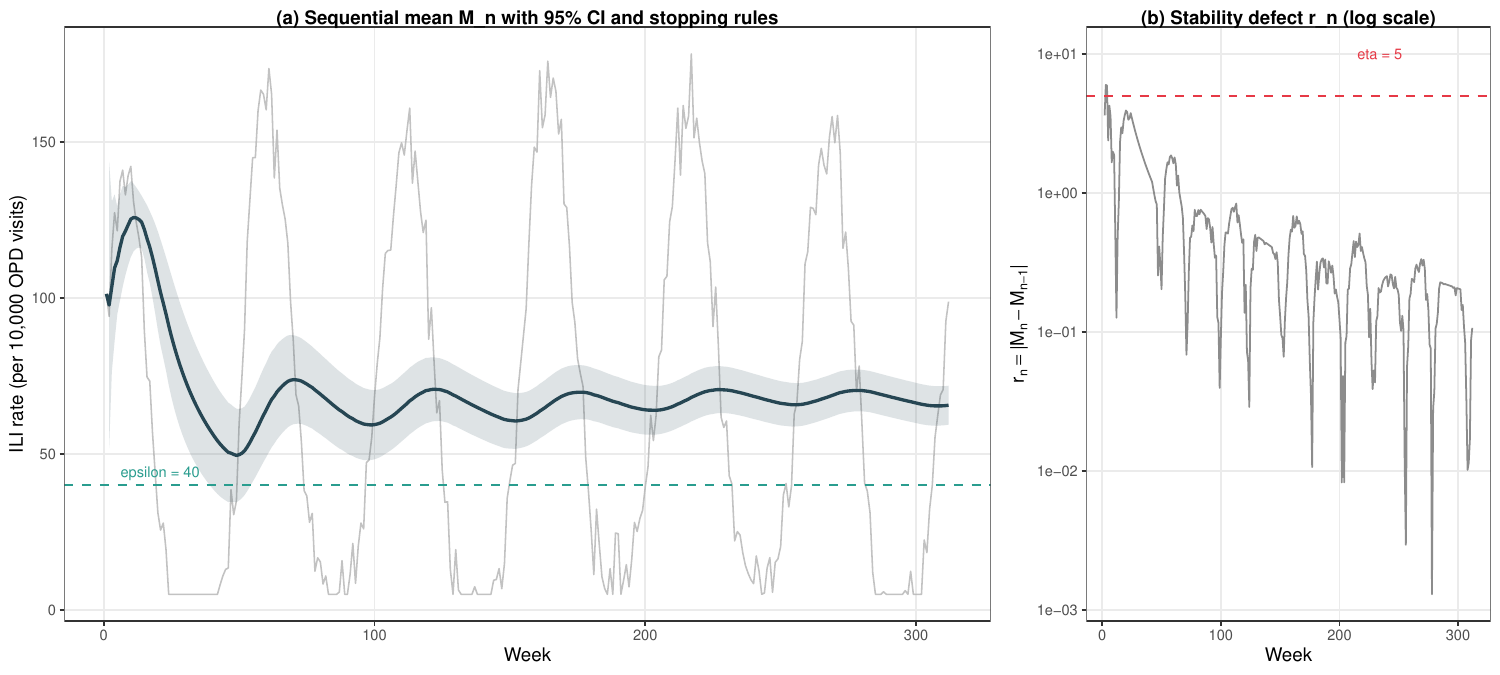}\\[4pt]
\includegraphics[width=0.75\linewidth]{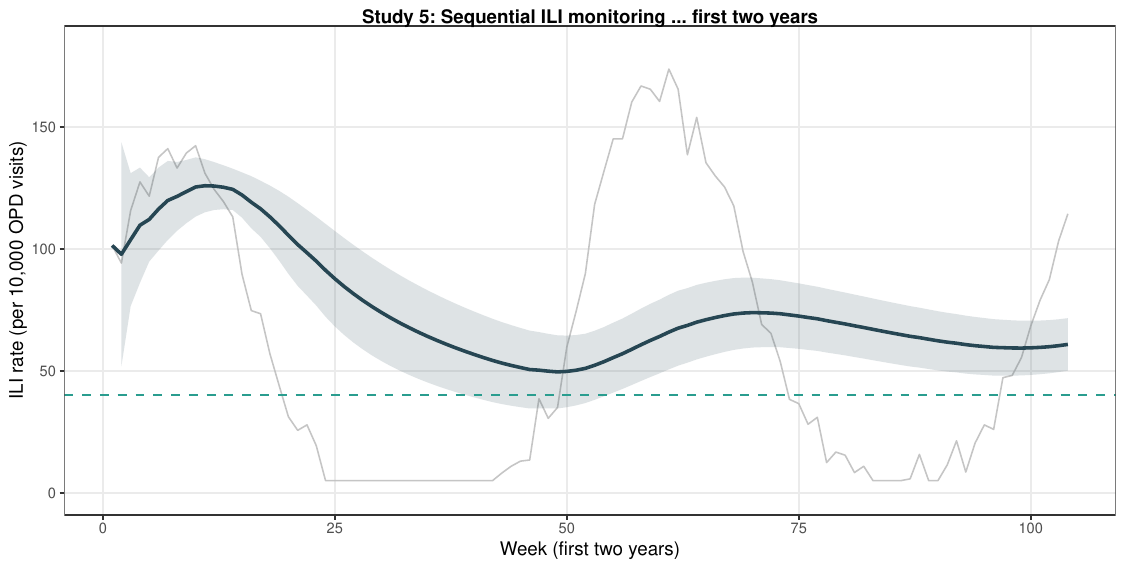}
\caption{Realistic-scale monitoring illustration calibrated to Taiwan CDC
influenza-like-illness (ILI) reporting patterns
(\(n=312\) weeks, synthetic series).
\textit{Upper:} sequential running mean \(M_n\) with a \(95\%\) posterior
credible band (Gamma prior), practical threshold \(\varepsilon\), and
stability defect \(r_n\) on a secondary axis.  \textit{Lower:} zoomed view
of the first two years (104 weeks).  The process never simultaneously
satisfies boundary closeness, uncertainty localization, and trajectory
stability, so no practical boundary declaration is made.  The persistent
non-zero defect $r_n$ reflects the information leakage intrinsic to seasonal
ILI reporting: week-to-week rate fluctuations prevent the conditional target
from stabilizing at a boundary.  This illustrates why trajectory stability
is a substantive requirement rather than a technicality, and why the diagnostic
is especially informative when run on the authentic historical series.}
\label{fig:study05}
\end{figure}

Figure~\ref{fig:study06} presents an NHANES-style blood-lead monitoring
illustration.  In the displayed sequential summary, the running mean falls
below the practical level \(\varepsilon=1.50\) quickly enough that the
boundary-only rule stops at \(n=50\), but \(\tau_{\mathrm{RM}}\) waits
until \(n=163\).  The empirical distribution function shows why this extra
caution is sensible: the threshold is not a zero-probability boundary point
but an interior practical cutoff, so boundary claims should be based on
stabilized evidence rather than a single early crossing.

\begin{figure}[!ht]
\centering
\includegraphics[width=0.95\linewidth]{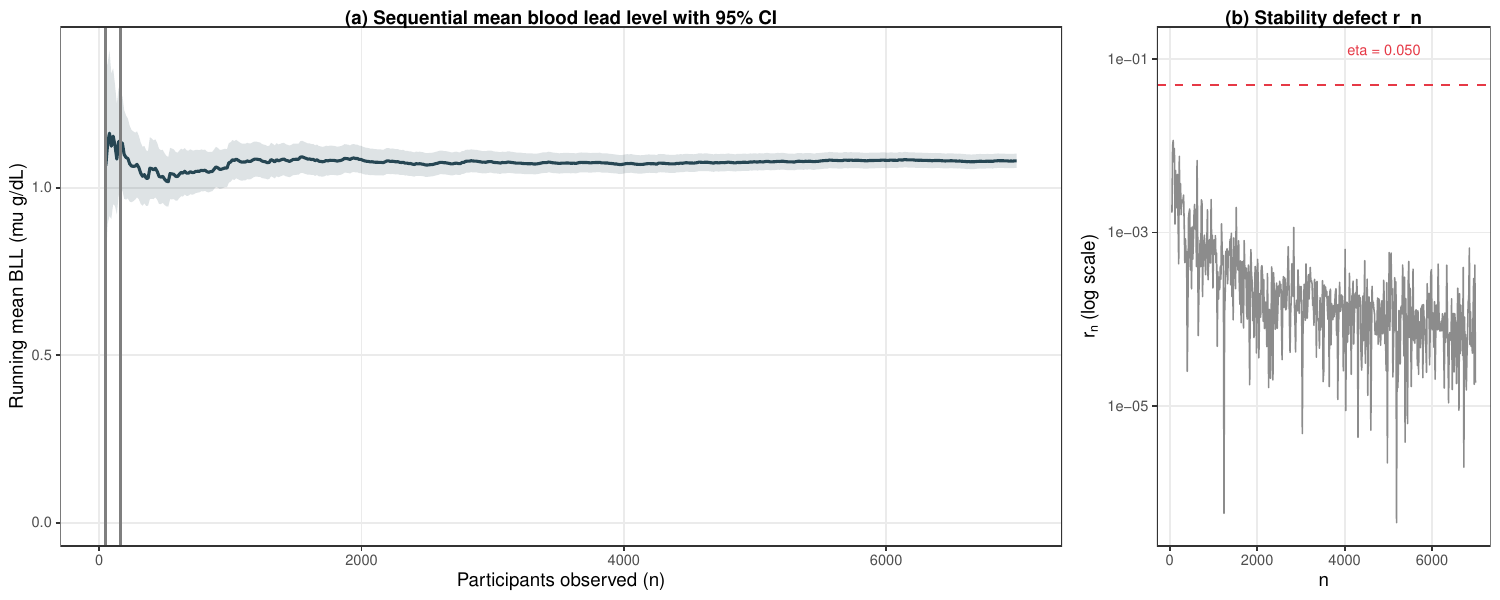}\\[4pt]
\includegraphics[width=0.52\linewidth]{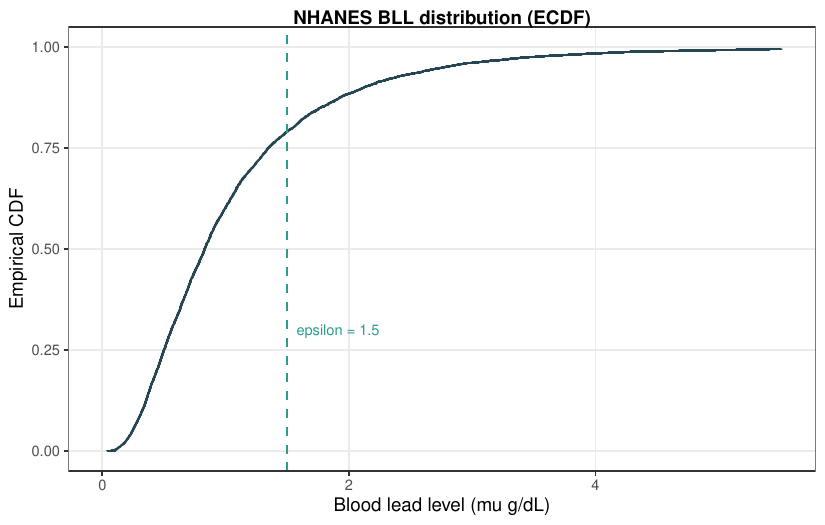}
\caption{Realistic-scale monitoring illustration calibrated to NHANES
2017--2018 blood-lead-level (BLL) measurements \citep{nhanes2017}
(synthetic series: \(n=7{,}000\), log-normal calibration,
geometric mean \(\approx 0.82\,\mu\mathrm{g/dL}\)).
\textit{Upper:} sequential running mean \(M_n\) with \(95\%\) Jeffreys
credible band.  The boundary-only rule \(\tau_{\mathrm{bdy}}\) fires at
\(n=50\) (triangle marker), while the full scorecard \(\tau_{\mathrm{RM}}\)
waits until \(n=163\) (circle marker), when boundary closeness, uncertainty
control, and trajectory stability are simultaneously satisfied.
\textit{Lower:} empirical CDF of the synthetic BLL series, emphasizing
that \(\varepsilon=1.50\,\mu\mathrm{g/dL}\) is an interior benchmark rather
than a point mass.  On the real NHANES survey series, the stability defect
$r_n$ further encodes the representation mismatch arising from the complex
sampling design and log-normal heteroscedasticity of BLL measurements, making
the stability requirement substantively important rather than nominal.}
\label{fig:study06}
\end{figure}

\subsection{Summary of numerical results}
\label{subsec:num_summary}

Taken together, the numerical studies support a single interpretation.  A value
near zero or near one is not, by itself, a reliable inferential statement.  The
same finite sample can look boundary-like because of an all-failure prefix, a
separating logistic fit, a transient surveillance dip, or an approximately
coherent compression scheme.  The reverse-martingale scorecard separates such
cases from genuinely stable boundary evidence by requiring three ingredients
simultaneously: closeness \(B_n\), uncertainty localization \(W_n\), and
trajectory stability \(r_n\).  This is the distinction needed to connect
classical sufficiency, modern sequential monitoring, and practical boundary
degeneracy within a single target-oriented compression framework.

In the exact sufficiency examples considered here, namely the Bernoulli,
Gaussian, and Poisson cases, the implemented diagnostic is automatically
nonbinding.  Thus the structural prediction of
Proposition~\ref{prop:exact_reverse_coherence} is confirmed: the stopping rules
\(\tau_{\mathrm{RM}}\) and \(\tau_{2\mathrm{cond}}\) are identical, so the
stability screen imposes no additional delay.  In the logistic study, the third condition \(r_n\leq\eta\) provides additional
protection relative to the two-condition rule.  In the quasi-martingale
perturbation study, the defect process is primarily diagnostic under the
present calibration, and would become decision-relevant under a stricter
stability threshold.

\section{Conclusion}

This article has developed a unified view of statistical information as
target-oriented compression.  The central point is that a statistic, by
itself, is not the natural martingale object.  Rather, a statistic induces a
retained information field, and the corresponding conditional target process
\[
    M_n = \E(Z \mid \G_n)
\]
is the object to which reverse-martingale theory applies when the retained
fields are arranged as a decreasing filtration.  This distinction clarifies
the roles of statistics, \(\sigma\)-fields, and conditional expectations:
statistics compress data, \(\sigma\)-fields encode retained information, and
reverse martingales describe the stable projection of a target under
successive information reduction.

Classical sufficiency appears in this framework as the lossless case of
statistical compression.  When a sufficient statistic preserves all information
relevant to a model parameter, the conditional target process based on its
\(\sigma\)-field retains the appropriate inferential content.  Approximate
summaries, including selected features, penalized estimates, risk scores, and
learned hidden states, need not satisfy exact reverse-martingale identities.
Their departure from coherence can instead be described through reverse
quasi-martingale defects.  These defects provide a quantitative language for
discussing how much target-relevant information is lost or distorted as data
are compressed.

The same perspective also resolves a recurring difficulty in sequential binary
inference: finite-sample boundary closeness should not be confused with exact
boundary degeneracy.  An all-failure Bernoulli run, a separated logistic
regression, or a near-zero sequential risk estimate may suggest a boundary
probability, but finite data alone do not justify the statement that a
probability is exactly zero or one.  The applied stopping procedures
formalizing this principle are studied in detail in the companion paper
\citet{chang2025rm}; the present paper provides their theoretical foundation
in the compression framework.  Exact degeneracy, when it occurs, is a limiting
property read from
\[
    M_\infty = \E(Y \mid \G_\infty),
\]
whereas practical boundary claims require additional safeguards.  The proposed
rule therefore requires three simultaneous conditions: closeness to the
boundary, uncertainty control, and trajectory stability.  This separates
transient numerical boundary behavior from stable inferential evidence.

The numerical studies support this distinction across rare-event Bernoulli
sampling, logistic separation, Gaussian and Poisson monitoring, data-scale
sequential illustrations, and quasi-reverse-martingale perturbations.  In
these examples, boundary-only rules often react quickly to early or unstable
sample-path behavior, whereas the full reverse-martingale scorecard is more
conservative because it requires the near-boundary state to be statistically
localized and dynamically stable.  The empirical message is consistent with
the theory: near zero or near one is not, by itself, a reliable scientific
conclusion.

Several extensions remain natural.  One direction is to develop sharper
finite-sample confidence sequences specifically for limiting
reverse-martingale targets.  Another is to study data-adaptive choices of the
stability threshold \(\eta\) in high-dimensional and learned-representation
settings.  A third is to connect the reverse-coherence defect more directly
with regularization, representation learning, and dynamic treatment rules.
More broadly, the framework suggests that modern statistical procedures should
be evaluated not only by predictive accuracy or likelihood fit, but also by
whether their compressed representations preserve the target-relevant
information needed for stable decisions.

In summary, the reverse-martingale compression view provides a common language
for sufficiency, likelihood-based estimation, approximate representation,
sequential monitoring, and practical boundary degeneracy.  Reliable boundary
conclusions require more than an extreme finite estimate: they require a stable
conditional target, localized uncertainty, and coherent behavior across the
compression path.

\medskip
\noindent\textbf{Supplementary Material.}
The Supplementary Material contains: (S1) proofs of all reverse quasi-martingale
defect bounds and the summability criterion for $\sum_n\E|\delta_n|<\infty$;
(S2) finite-sample error-bound derivations for $\tau_{\mathrm{RM}}$;
(S3) all R scripts for reproducing the numerical studies; and (S4) synthetic
data files for Studies~5--6 together with instructions for substituting the
actual public-domain NHANES and Taiwan CDC series.

\appendix
\section*{Appendix}
\section{Classical Bernoulli benchmarks}
\label{app:benchmarks}

The all-failure practical-zero rule (Section~\ref{sec:rm_boundary}) is the
\(s=0\) special case of the exact one-sided binomial test of
\(H_0:p\geq\varepsilon\) against \(H_1:p<\varepsilon\).  Its exact
\(p\)-value at the boundary is
\[
    (1-\varepsilon)^n,
\]
and the dual Clopper--Pearson upper bound \citep{clopper1934} after zero
successes is
\[
    U_n(0;\alpha)=1-\alpha^{1/n}.
\]
For simple-versus-simple Bernoulli monitoring, Wald's sequential probability
ratio test \citep{wald1945,wald1947,waldwolfowitz1948,siegmund1985} compares \(H_0:p=p_0\) with \(H_1:p=p_1\), \(p_1<p_0\), using
\[
    L_n
    =
    S_n\log\frac{p_1}{p_0}
    +(n-S_n)\log\frac{1-p_1}{1-p_0}.
\]
With Wald boundaries
\[
    a=\log\frac{\beta}{1-\alpha},
    \qquad
    b=\log\frac{1-\beta}{\alpha},
\]
an all-failure path favors \(H_1\) after
\[
    \tau_{\mathrm{SPRT},0}
    =
    \left\lceil
    \frac{\log\{(1-\beta)/\alpha\}}
         {\log\{(1-p_1)/(1-p_0)\}}
    \right\rceil
\]
observations.  This benchmark is appropriate when the scientific problem is a
choice between two pre-specified point hypotheses.  The reverse-martingale
boundary framework addresses the complementary setting of composite practical
boundary statements, logistic separation, and trajectory stability.

Confidence sequences \citep{robbins1970,howard2021} provide the time-uniform uncertainty
component \(W_n\) in the unified rule of Section~\ref{sec:stopping}.  If \(C_n\) satisfies
\[
    \Pp_p\{p\in C_n\text{ for all }n\geq1\}\geq1-\alpha,
\]
then a practical-zero declaration based on \(C_n\subseteq[0,\varepsilon]\) can
be embedded in \(\tau_{\mathrm{RM}}\) by taking \(W_n\) to be the width of
\(C_n\) and by requiring the interval to lie inside the chosen practical
boundary region.

\section{Numerical settings summarized}
\label{app:numerical_settings}

The executed numerical studies used the following designs
(all scripts in the accompanying \texttt{scripts/} directory,
reproducible via \texttt{Rscript run\_all.R}):
\begin{enumerate}[label=(\arabic*)]
    \item \textit{Bernoulli rare-event} (Study~1): \(Y_i\iid\mathrm{Bernoulli}(p)\),
    \(p\in\{0.001,0.005,0.010,0.050\}\), \(\varepsilon\in\{0.005,0.010\}\),
    \(N_{\max}=5{,}000\), \(B=1{,}000\), \(w=0.02\), \(\eta=10^{-6}\).
    Posterior width from Jeffreys Beta\((1/2,1/2)\) prior.
    (\emph{Note on \(\eta\):} for exact-sufficient summaries the stability
    defect satisfies \(r_n<10^{-10}\) numerically; setting
    \(\eta=10^{-6}\) is effectively machine-precision and imposes no
    additional delay beyond the uncertainty condition.  This is in line with
    Proposition~\ref{prop:exact_reverse_coherence}.)
    \item \textit{Logistic regression} (Study~2): \(Y_i\given x_i\sim\mathrm{Bernoulli}(\mathrm{expit}(x_i^\top\beta))\),
    \(x_i\sim N(0,I_d)\), true \(\beta=({\rm logit}(\rho),0,\ldots,0)^\top\);
    scenarios \(d\in\{3,20\}\), \(\rho\in\{0.005,0.010\}\);
    ridge logistic regression with \(\lambda=1\);
    \(\varepsilon=w=0.05\), \(\eta=0.01\), \(N_{\max}\in\{500,1{,}000\}\).
    (\emph{Note on \(\eta\):} the ridge-MLE predictive surface produces
    step-to-step changes \(r_n\) of order \(10^{-2}\) to \(10^{-3}\) near
    stopping; \(\eta=0.01\) is therefore a practically loose threshold
    calibrated to the typical signal size in the approximate-compression
    regime.  A stricter \(\eta\) would increase delays further without
    materially changing the qualitative findings.)
    \item \textit{Normal calibration} (Study~3): \(X_i\iid N(\mu,1)\),
    \(\mu\in\{0,0.01,0.02,0.05,0.10\}\);
    Jeffreys-like \(N(0,1)\) prior; CUSUM calibrated for
    \(\mathrm{ARL}_0=500\); \(\varepsilon=w=0.05\), \(N_{\max}=3{,}000\).
    \item \textit{Poisson surveillance} (Study~4):
    \(X_i\iid\mathrm{Poisson}(\lambda)\),
    \(\lambda\in\{0.001,0.005,0.010,0.050\}\),
    \(\varepsilon\in\{0.005,0.010\}\);
    Jeffreys Gamma\((1/2)\) prior; CUSUM and SPRT benchmarks;
    \(N_{\max}=5{,}000\), \(w=0.02\).
    \item \textit{Taiwan CDC ILI illustration} (Study~5):
    weekly influenza-like illness rate; \(n=312\) weeks
    (2018-W01 through 2023-W52); synthetic series calibrated to
    Taiwan CDC seasonal reporting patterns when
    the local data file is absent.
    See Section~\ref{app:data_sources} for the data source and
    reproducibility instructions.
    \item \textit{NHANES blood-lead illustration} (Study~6):
    blood-lead level \(\mu\mathrm{g/dL}\); \(n=7{,}000\) observations;
    synthetic log-normal series (geometric mean \(\approx 0.82\,\mu\mathrm{g/dL}\),
    geometric SD \(\approx 2.1\)) calibrated to the NHANES 2017--2018
    Laboratory Procedures Manual \citep{nhanes2017};
    practical threshold \(\varepsilon=1.50\,\mu\mathrm{g/dL}\).
    See Section~\ref{app:data_sources} for the data source and
    reproducibility instructions.
    \item \textit{Quasi-reverse-martingale perturbations} (Study~7):
    three scenarios (misspecification, exponential smoothing, VB dampening)
    with nine parameter variants; \(B=1{,}000\), \(n_{\max}=2{,}000\).
\end{enumerate}
Section~\ref{sec:discussion} incorporates the numerical-study figures and
tables directly from the bundled output files so that the empirical diagnostics
can be read alongside the theory.

\subsection{Data sources}
\label{app:data_sources}

Studies~1--4 and~7 use fully simulated data generated within the accompanying
R scripts; no external data files are required.  Studies~5 and~6 are
calibrated to publicly available administrative and survey data as described
below.

\paragraph{Study~5: Taiwan CDC influenza-like illness data.}
The intended real data series is the weekly influenza-like illness (ILI)
consultation rate per 10,000 outpatient visits reported by the Taiwan Centers
for Disease Control (TCDC) through the National Notifiable Disease Surveillance
System (NNDSS), covering 2018-W01 through 2023-W52 (\(n=312\) weeks).  This
series is freely accessible at \url{https://nidss.cdc.gov.tw/}.  Because the
data file \texttt{scripts/data/taiwan\_ili\_weekly.csv} (columns:
\texttt{year}, \texttt{week}, \texttt{rate}) was not present at the time of
execution, the script \texttt{study10\_taiwan\_cdc.R} fell back to a synthetic
series whose seasonal structure (annual mean $\approx 60$, epidemic peak
$\approx 150$--$200$, off-season trough $\approx 15$--$25$ cases per 10,000
outpatient visits) was calibrated to published TCDC surveillance summaries.  Placing
the downloaded file at the path above causes the script to use the real data
automatically on any subsequent run.  The methodological conclusions do not
depend on access to the original administrative records.

\paragraph{Study~6: NHANES blood-lead data.}
The intended real data are blood-lead level (BLL, $\mu$g/dL) measurements from
the National Health and Nutrition Examination Survey (NHANES) 2017--2018 cycle,
laboratory component Blood Metals, file \texttt{PBCD\_J}, variable
\texttt{LBXBPB}.  The data are freely available at
\url{https://wwwn.cdc.gov/Nchs/Nhanes/2017-2018/PBCD_J.XPT} and can be
loaded directly in R via \texttt{nhanesA::nhanes("PBCD\_J")} (CRAN package
\texttt{nhanesA}).  Because this download was unavailable at execution time,
the script \texttt{study11\_nhanes.R} generated a synthetic log-normal series
(\(n=7{,}000\), geometric mean $\approx 0.82\,\mu$g/dL, geometric SD
$\approx 2.1$, detection floor $0.01\,\mu$g/dL) calibrated to the summary
statistics documented in the NHANES 2017--2018 Laboratory Procedures Manual
\citep{nhanes2017}.  The
synthetic file is saved as
\texttt{scripts/data/nhanes\_bll\_synthetic.csv} and is included in the
supplementary package.  Placing a real-data file at
\texttt{scripts/data/nhanes\_bll\_2017\_2018.csv} causes the script to use
the survey data automatically.  The NHANES data are in the public domain
(U.S. federal government work).

\section{Sufficiency derivations for Example~\ref{ex:exp_family}}
\label{app:sufficiency_derivations}

\paragraph{Normal, known variance.}
For $X_i\iid N(\mu,\sigma^2)$ with $\sigma^2$ known, expand
$\sum_i(x_i-\mu)^2 = \sum_i x_i^2 - 2\mu\sum_i x_i + n\mu^2$.
The likelihood factors as $L(\mu;x)=g_\mu(\bar x_n)\cdot h(x)$ with
$g_\mu(\bar x_n)\propto\exp\{n\bar x_n\mu/\sigma^2 - n\mu^2/(2\sigma^2)\}$
and $h(x)=\exp\{-\sum_i x_i^2/(2\sigma^2)\}$, confirming that $\bar X_n$
is sufficient by the Fisher--Neyman factorization \citep{fisher1922,lehmanncasella1998}.

\paragraph{Bernoulli.}
For $X_i\iid\mathrm{Bernoulli}(p)$,
$L(p;x)=p^{s_n}(1-p)^{n-s_n}$ where $s_n=\sum_i x_i$.
This factors through $T_n=s_n$ alone, so $T_n$ is sufficient by the
Fisher--Neyman factorization.  Under a Jeffreys $\mathrm{Beta}(1/2,1/2)$
prior, the posterior is $\mathrm{Beta}(s_n+1/2,\,n-s_n+1/2)$ and the
posterior mean is $M_n=(s_n+1/2)/(n+1)$.

The reverse-martingale property $\E(M_n\mid\G_{n+1})=M_{n+1}$ — where
$\G_{n+1}=\sigma(T_{n+1},T_{n+2},\ldots)$ is the tail $\sigma$-field at
step $n+1$ — follows immediately from Theorem~\ref{thm:conditional_rm}: since
$M_k=\E(Z\mid\G_k)$ for all $k$ and $(\G_k)$ is a decreasing filtration,
\[
  \E(M_n\mid\G_{n+1})
  =\E\bigl[\E(Z\mid\G_n)\mid\G_{n+1}\bigr]
  =\E(Z\mid\G_{n+1})
  =M_{n+1},
\]
where the second equality uses the tower property together with
$\G_{n+1}\subseteq\G_n$.  The explicit posterior-mean formula confirms this:
$M_n=(s_n+1/2)/(n+1)$ with $s_n=T_n$ and $M_{n+1}=(s_{n+1}+1/2)/(n+2)$
with $s_{n+1}=T_{n+1}$, and one verifies
$\E(M_n\mid\G_{n+1})=\E[(s_n+1/2)/(n+1)\mid s_{n+1}]
=(s_{n+1}+1/2)/(n+2)=M_{n+1}$
by the identity $\E[s_n\mid s_{n+1}]=s_{n+1}\cdot n/(n+1)$ for the
conditional distribution of a sub-sum given the total in the Bernoulli
model.

\paragraph{Poisson.}
For $X_i\iid\mathrm{Poisson}(\lambda)$,
$L(\lambda;x)\propto e^{-n\lambda}\lambda^{s_n}$, so $T_n=\sum_i X_i$ is
sufficient.  Under a Jeffreys $\mathrm{Gamma}(1/2,1)$ prior, the posterior
is $\mathrm{Gamma}(s_n+1/2, n+1)$ and the posterior mean is
$M_n=(s_n+1/2)/(n+1)$.  The reverse-martingale property follows analogously.

\section*{Data Availability Statement}

All simulation scripts and synthetic data files are included in the
supplementary \texttt{scripts/} directory and are reproducible by running
\texttt{Rscript run\_all.R}.  Studies~1--4 and~7 require no external data.
The Taiwan CDC ILI series used in Study~5 is publicly available at
\url{https://nidss.cdc.gov.tw/}.  The NHANES blood-lead series used in
Study~6 is publicly available at
\url{https://wwwn.cdc.gov/Nchs/Nhanes/2017-2018/PBCD_J.XPT} and in the
public domain.  Both studies include fully reproducible synthetic fallback
series calibrated to the respective public sources; the substantive
methodological conclusions do not depend on access to the original
administrative or survey records.


\begin{thebibliography}{99}

\bibitem[Albert and Anderson(1984)]{albert1984}
Albert, A. and Anderson, J.\,A. (1984).
On the existence of maximum likelihood estimates in logistic regression models.
\textit{Biometrika}, 71(1), 1--10.
\href{https://doi.org/10.1093/biomet/71.1.1}{doi:10.1093/biomet/71.1.1}

\bibitem[Björk and Johansson(1996)]{bjork1996}
Björk, T. and Johansson, B. (1996).
Parameter estimation and reverse martingales.
\textit{Stochastic Processes and their Applications}, 63(2), 235--263.
\href{https://doi.org/10.1016/0304-4149(96)00080-4}{doi:10.1016/0304-4149(96)00080-4}

\bibitem[Chakraborty and Moodie(2013)]{chakraborty2013}
Chakraborty, B. and Moodie, E.\,E.\,M. (2013).
\textit{Statistical Methods for Dynamic Treatment Regimes}.
Springer, New York.
\href{https://doi.org/10.1007/978-1-4614-7428-9}{doi:10.1007/978-1-4614-7428-9}

\bibitem[Chang(2026)]{chang2025rm}
Chang, Y.-c.\,I. (2026).
Practical boundary degeneracy and reverse-martingale limits in sequential
binary models.
Preprint, \href{https://arxiv.org/abs/2605.02274}{arXiv:2605.02274 [stat.ME]}.

\bibitem[Clopper and Pearson(1934)]{clopper1934}
Clopper, C.\,J. and Pearson, E.\,S. (1934).
The use of confidence or fiducial limits illustrated in the case of the binomial.
\textit{Biometrika}, 26(4), 404--413.
\href{https://doi.org/10.1093/biomet/26.4.404}{doi:10.1093/biomet/26.4.404}

\bibitem[Doob(1953)]{doob1953}
Doob, J.\,L. (1953).
\textit{Stochastic Processes}.
John Wiley \& Sons, New York.

\bibitem[Durrett(2019)]{durrett2019}
Durrett, R. (2019).
\textit{Probability: Theory and Examples} (5th ed.).
Cambridge University Press, Cambridge.
\href{https://doi.org/10.1017/9781108591034}{doi:10.1017/9781108591034}

\bibitem[Firth(1993)]{firth1993}
Firth, D. (1993).
Bias reduction of maximum likelihood estimates.
\textit{Biometrika}, 80(1), 27--38.
\href{https://doi.org/10.1093/biomet/80.1.27}{doi:10.1093/biomet/80.1.27}

\bibitem[Fisher(1922)]{fisher1922}
Fisher, R.\,A. (1922).
On the mathematical foundations of theoretical statistics.
\textit{Philosophical Transactions of the Royal Society A}, 222, 309--368.

\bibitem[Fong et~al.(2023)]{fong2023}
Fong, E., Holmes, C., and Walker, S.\,G. (2023).
Martingale posterior distributions.
\textit{Journal of the Royal Statistical Society: Series B}, 85(5), 1357--1391.
\href{https://doi.org/10.1093/jrsssb/qkad005}{doi:10.1093/jrsssb/qkad005}

\bibitem[Gelman et~al.(2008)]{gelman2008}
Gelman, A., Jakulin, A., Pittau, M.\,G., and Su, Y.-S. (2008).
A weakly informative default prior distribution for logistic and other regression models.
\textit{The Annals of Applied Statistics}, 2(4), 1360--1383.
\href{https://doi.org/10.1214/08-AOAS191}{doi:10.1214/08-AOAS191}

\bibitem[Goodfellow et~al.(2016)]{goodfellow2016}
Goodfellow, I., Bengio, Y., and Courville, A. (2016).
\textit{Deep Learning}.
MIT Press, Cambridge, MA.

\bibitem[Heinze and Schemper(2002)]{heinze2002}
Heinze, G. and Schemper, M. (2002).
A solution to the problem of separation in logistic regression.
\textit{Statistics in Medicine}, 21(16), 2409--2419.
\href{https://doi.org/10.1002/sim.1047}{doi:10.1002/sim.1047}

\bibitem[Hochreiter and Schmidhuber(1997)]{hochreiter1997}
Hochreiter, S. and Schmidhuber, J. (1997).
Long short-term memory.
\textit{Neural Computation}, 9(8), 1735--1780.

\bibitem[Howard et~al.(2021)]{howard2021}
Howard, S.\,R., Ramdas, A., McAuliffe, J., and Sekhon, J. (2021).
Time-uniform, nonparametric, nonasymptotic confidence sequences.
\textit{The Annals of Statistics}, 49(2), 1055--1080.
\href{https://doi.org/10.1214/20-AOS1991}{doi:10.1214/20-AOS1991}

\bibitem[Kallenberg(2002)]{kallenberg2002}
Kallenberg, O. (2002).
\textit{Foundations of Modern Probability} (2nd ed.).
Springer, New York.

\bibitem[Lehmann and Casella(1998)]{lehmanncasella1998}
Lehmann, E.\,L. and Casella, G. (1998).
\textit{Theory of Point Estimation} (2nd ed.).
Springer, New York.

\bibitem[Murphy(2003)]{murphy2003}
Murphy, S.\,A. (2003).
Optimal dynamic treatment regimes.
\textit{Journal of the Royal Statistical Society: Series B}, 65(2), 331--355.
\href{https://doi.org/10.1111/1467-9868.00389}{doi:10.1111/1467-9868.00389}

\bibitem[NCHS(2020)]{nhanes2017}
National Center for Health Statistics (2020).
NHANES 2017--2018: Laboratory Procedures Manual.
National Center for Health Statistics, Centers for Disease
Control and Prevention, U.S. Department of Health and Human Services,
Hyattsville, MD.
\url{https://wwwn.cdc.gov/Nchs/Nhanes/2017-2018/PBCD_J.htm}

\bibitem[Robins(2004)]{robins2004}
Robins, J.\,M. (2004).
Optimal structural nested models for optimal sequential decisions.
In D.\,Y. Lin and P.\,J. Heagerty (eds.),
\textit{Proceedings of the Second Seattle Symposium in Biostatistics},
Lecture Notes in Statistics, vol. 179, pp. 189--326.
Springer, New York.
\href{https://doi.org/10.1007/978-1-4419-9076-1_11}{doi:10.1007/978-1-4419-9076-1\_11}

\bibitem[Robbins(1970)]{robbins1970}
Robbins, H. (1970).
Statistical methods related to the law of the iterated logarithm.
\textit{The Annals of Mathematical Statistics}, 41(5), 1397--1409.
\href{https://doi.org/10.1214/aoms/1177696786}{doi:10.1214/aoms/1177696786}

\bibitem[Siegmund(1985)]{siegmund1985}
Siegmund, D. (1985).
\textit{Sequential Analysis: Tests and Confidence Intervals}.
Springer, New York.
\href{https://doi.org/10.1007/978-1-4613-9549-7}{doi:10.1007/978-1-4613-9549-7}

\bibitem[Ville(1939)]{ville1939}
Ville, J. (1939).
\textit{{\'E}tude Critique de la Notion de Collectif}.
Gauthier-Villars, Paris.

\bibitem[Wald(1945)]{wald1945}
Wald, A. (1945).
Sequential tests of statistical hypotheses.
\textit{The Annals of Mathematical Statistics}, 16(2), 117--186.
\href{https://doi.org/10.1214/aoms/1177731118}{doi:10.1214/aoms/1177731118}

\bibitem[Wald(1947)]{wald1947}
Wald, A. (1947).
\textit{Sequential Analysis}.
John Wiley \& Sons, New York.

\bibitem[Wald and Wolfowitz(1948)]{waldwolfowitz1948}
Wald, A. and Wolfowitz, J. (1948).
Optimum character of the sequential probability ratio test.
\textit{The Annals of Mathematical Statistics}, 19(3), 326--339.
\href{https://doi.org/10.1214/aoms/1177730197}{doi:10.1214/aoms/1177730197}

\bibitem[Waudby-Smith and Ramdas(2023)]{waudbysmith2023}
Waudby-Smith, I. and Ramdas, A. (2023).
Estimating means of bounded random variables by betting.
\textit{Journal of the Royal Statistical Society: Series B}, 85(1), 1--26.
\href{https://doi.org/10.1093/jrsssb/qkac007}{doi:10.1093/jrsssb/qkac007}

\bibitem[Williams(1991)]{williams1991}
Williams, D. (1991).
\textit{Probability with Martingales}.
Cambridge University Press, Cambridge.

\end{thebibliography}
\end{document}